\documentclass[11pt]{aart}

\usepackage[letterpaper, hmargin=1in, top=1.1in, bottom=1.3in, footskip=0.7in]{geometry}

\usepackage{titlesec}
\titleformat{\subsection}[runin]{\normalfont\bfseries}{\thesubsection.}{.5em}{}[.]\titlespacing{\subsection}{0pt}{2ex plus .1ex minus .2ex}{.8em}
\titleformat{\subsubsection}[runin]{\normalfont\itshape}{\thesubsubsection.}{.3em}{}[.]\titlespacing{\subsubsection}{0pt}{1ex plus .1ex minus .2ex}{.5em}
\titleformat{\paragraph}[runin]{\normalfont\itshape}{\theparagraph.}{.3em}{}[.]\titlespacing{\paragraph}{0pt}{1ex plus .1ex minus .2ex}{.5em}

\usepackage[labelfont=sc,font=small,labelsep=period]{caption}
\usepackage{amsmath}
\usepackage{amssymb}
\usepackage{amsfonts}
\usepackage{latexsym}
\usepackage{amsthm}
\usepackage{amsxtra}
\usepackage{amscd}
\usepackage{bbm}
\usepackage{mathrsfs}
\usepackage{bm}

\usepackage{graphicx, color}

\definecolor{darkred}{rgb}{0.9,0,0.3}
\definecolor{darkblue}{rgb}{0,0.3,0.9}

\definecolor{vdarkred}{rgb}{0.6,0,0.2}
\definecolor{vdarkblue}{rgb}{0,0.2,0.6}
\usepackage[pdftex, colorlinks, linkcolor=vdarkblue,citecolor=vdarkred]{hyperref}

\usepackage[nottoc,notlof,notlot]{tocbibind}
\usepackage{cite} 

\numberwithin{equation}{section}
\numberwithin{figure}{section}

\theoremstyle{plain} 
\newtheorem{theorem}{Theorem}[section]
\newtheorem*{theorem*}{Theorem}
\newtheorem{lemma}[theorem]{Lemma}
\newtheorem*{lemma*}{Lemma}

\newtheorem*{corollary*}{Corollary}
\newtheorem{proposition}[theorem]{Proposition}
\newtheorem*{proposition*}{Proposition}

\newtheorem*{conjecture*}{Conjecture}

\theoremstyle{definition} 
\newtheorem{definition}[theorem]{Definition}
\newtheorem*{definition*}{Definition}

\newtheorem*{example*}{Example}
\newtheorem{remark}[theorem]{Remark}
\newtheorem*{remark*}{Remark}

\newtheorem*{assumption*}{Assumption}

\renewcommand{\cal}{\mathcal} 
 
\newcommand{\fra}{\mathfrak} 

\renewcommand{\P}{\mathbb{P}}

\newcommand{\R}{\mathbb{R}}
\newcommand{\C}{\mathbb{C}}
\newcommand{\N}{\mathbb{N}}
\newcommand{\Z}{\mathbb{Z}}

\newcommand{\ee}{\mathrm{e}}
\newcommand{\ii}{\mathrm{i}}
\newcommand{\dd}{\mathrm{d}}

\newcommand*{\deq}{\mathrel{\vcenter{\baselineskip0.65ex \lineskiplimit0pt \hbox{.}\hbox{.}}}=}

\renewcommand{\leq}{\leqslant}
\renewcommand{\geq}{\geqslant}
\renewcommand{\epsilon}{\varepsilon}

\newcommand{\pb}[1]{\bigl({#1}\bigr)}

\newcommand{\abs}[1]{\lvert #1 \rvert}

\newcommand{\norm}[1]{\lVert #1 \rVert}

\newcommand{\scalar}[2]{\langle{#1} \mspace{2mu}, {#2}\rangle}

\DeclareMathOperator{\tr}{Tr}

\DeclareMathOperator{\re}{Re}
\DeclareMathOperator{\im}{Im}

\DeclareMathOperator{\spec}{spec}

\begin{document}

\title{A microscopic derivation of time-dependent correlation functions of the $1D$ cubic nonlinear Schr\"{o}dinger equation}
\author{
J\"urg Fr\"ohlich\footnote{ETH Z\"urich, Institute for Theoretical Physics, {\tt juerg@phys.ethz.ch}.}
\and Antti Knowles\footnote{University of Geneva, Department of Mathematics, {\tt antti.knowles@unige.ch}.}
\and Benjamin Schlein\footnote{University of Z\"urich, Department of Mathematics, {\tt benjamin.schlein@math.uzh.ch}.}
\and Vedran Sohinger\footnote{University of Z\"{u}rich, Department of Mathematics, {\tt vedran.sohinger@math.uzh.ch}.}
}

\maketitle

\begin{abstract}
We give a microscopic derivation of time-dependent correlation functions of the $1D$ cubic nonlinear Schr\"{o}dinger equation (NLS) from many-body quantum theory.
The starting point of our proof is \cite{FrKnScSo} on the time-independent problem and \cite{Knowles_Thesis} on the corresponding problem on a finite lattice. An important new obstacle in our analysis is the need to work with a cutoff in the number of particles, which breaks the Gaussian structure of the free quantum field and prevents the use of the Wick theorem. We overcome it by the use of means of complex analytic methods. Our methods apply to the nonlocal NLS with bounded convolution potential. In the periodic setting, we also consider the local NLS, arising from short-range interactions in the many-body setting. To that end, we need the dispersion of the NLS in the form of periodic Strichartz estimates in $X^{s,b}$ spaces.
\end{abstract}

\section{Setup and main result}

Let $\fra H$ be a Hilbert space, $H \in C^{\infty}(\fra H)$ a Hamiltonian function, and $\{\cdot,\cdot\}$ a Poisson bracket on $C^{\infty}(\fra H) \times C^{\infty}(\fra H)$. We can then define the Hamiltonian flow of $H$ on $\fra H$, which we denote by $u \mapsto S_t u$. Furthermore, we introduce the \emph{Gibbs measure} associated with the Hamiltonian $H$, defined as a probability measure $\P$ on $\fra H$ formally given by
\begin{equation}
\label{Gibbs measure}
\dd \P (u)\;\deq\; \frac{1}{Z} \ee^{-H(u)} \,\dd u\,,
\end{equation}
where $Z$ is a positive normalization constant and $\dd u$ is Lebesgue measure on $\fra H$ (whish is ill-defined if $\fra H$ is infinite-dimensional).
The problem of the construction of measures of the type \eqref{Gibbs measure} was first considered in the constructive quantum field theory literature, c.f.\ \cite{Glimm_Jaffe,Simon74} and the references therein, and later in \cite{LRS,McKean_Vaninsky1,McKean_Vaninsky2}.
In the context of nonlinear dispersive PDEs, the invariance of measures of the type \eqref{Gibbs measure} has been considered in the work of Bourgain \cite{B,Bourgain_ZS,B1,B2,B3} and Zhidkov \cite{Zhidkov}, and in the subsequent literature. An important application of the invariance is to obtain a substitute for a conservation law at low regularity which, in turn, allows us to construct solutions for random initial data of low regularity.
We refer the reader to the introduction of \cite{FrKnScSo} for a detailed overview and for further references.

Given $\P$ as in \eqref{Gibbs measure}, natural objects to consider are the associated \emph{time-dependent correlation functions}. More precisely, for $m \in \N$, times $t_1, \ldots,t_m \in \R$, and
functions $X^1,\ldots,X^m \in C^\infty(\fra H)$, we consider
\begin{equation}
\label{C_P}
Q_{\P}(X^1,\ldots,X^m;t_1,\ldots,t_m) \;\deq\; \int X^1(S_{t_1}u) \, \cdots \, X^m(S_{t_m}u)\,\dd \P(u)\,,
\end{equation}
the $m$-particle time-dependent correlation function associated with $H$.
The goal of this paper is a microscopic derivation of \eqref{C_P} from the corresponding many-body quantum objects in the case when the Hamiltonian flow is the flow of a cubic nonlinear Schr\"{o}dinger equation in one spatial dimension.
This is the time-dependent variant of the question previously considered in \cite{FrKnScSo,Lewin_Nam_Rougerie}.

We now set this up more precisely. Let us consider the spatial domain $\Lambda=\mathbb{T}^1$ or $\mathbb{R}$.
The \emph{one-particle space} is given by $\fra H \deq L^2(\Lambda; \C)$. The scalar product and norm on $\fra H$ are denoted by $\scalar{\cdot}{\cdot}_{\fra H}$ and $\|\cdot\|_{\fra H}$ respectively. We use the convention that $\scalar{\cdot}{\cdot}_{\fra H}$ is linear in the second argument.
We start from the \emph{one-body Hamiltonian}
\begin{equation} 
\label{Hamiltonian h1}
h \;\deq\; -\Delta + \kappa + v\,,
\end{equation}
for a \emph{chemical potential} $\kappa>0$ and a \emph{one-body potential} $v : \Lambda \to [0,+\infty)$.
This is a positive, self-adjoint densely defined operator on $\fra H$. Furthermore, we assume that $h$ has a compact resolvent and satisfies 
\begin{equation}
\label{Tr h^{s-1}}
\tr h^{-1} \;<\;\infty\,.
\end{equation}
In particular, we can take $v=0$ when $\Lambda=\mathbb{T}^1$. 
We write the spectral representation of $h$ as
\begin{equation}
\label{Hamiltonian h}
h \;=\; \sum_{k \in \N} \lambda_k u_k u_k^*\,.
\end{equation}
Here $\lambda_k>0$ are the eigenvalues and $u_k$ are the associated normalized eigenfunctions in $\fra H$ of the operator $h$. We consider an \emph{interaction potential} $w$ that satisfies
\begin{equation}
\label{definition_w}
w \in L^{\infty}(\Lambda)\,,\quad w \geq 0 \mbox{ pointwise}\,.
\end{equation}
The Hamilton function that we consider is
\begin{equation}
\label{H_definition} 
H(u) \;\deq\; \int_\Lambda \dd x \, \pb{\abs{\nabla u(x)}^2 + v(x) \abs{u(x)}^2} + \frac{1}{2} \int_\Lambda \dd x \, \dd y \, \abs{u(x)}^2 \, w(x - y) \, \abs{u(y)}^2\,,
\end{equation}
where $\dd x$ denotes the Lebesgue measure on $\Lambda$. We often abbreviate $\int_\Lambda \dd x \equiv \int \dd x$.
The space of fields $u: \Lambda \to \C$ generates a Poisson algebra where the Poisson bracket is given by
\begin{equation} \label{Poisson}
\{u(x),\bar{u}(y)\}\;=\;\ii \delta(x-y) \,, \qquad \{u(x),u(y)\}\;=\; \{\bar{u}(x),\bar{u}(y)\} \;=\;0\,.
\end{equation}
The Hamiltonian equation of motion associated with \eqref{H_definition}--\eqref{Poisson} is the nonlocal \emph{nonlinear Schr\"{o}dinger equation (NLS)}
\begin{equation}
\label{NLS1}
\ii \partial_t u(x) + (\Delta-\kappa)u(x) \;=\;v(x)\, u(x) +\int_\Lambda \dd y\, |u(y)|^2\,w(x-y)\,u(x)\,.
\end{equation}
In addition to \eqref{NLS1}, we also consider the local NLS
\begin{equation}
\label{localNLS1}
\ii \partial_t u(x) + (\Delta-\kappa)u(x) \;=\;|u(x)|^2u(x)\,,
\end{equation}
obtained from \eqref{NLS1} by setting $v=0$ and $w=\delta$.
This is the Hamiltonian equation of motion associated with the Hamiltonian obtained from \eqref{H_definition} by the analogous modifications. 

By the arguments of \cite{Bourgain_1993} we know that both \eqref{NLS1} and \eqref{localNLS1} are globally well-posed in $\fra H$.
Given initial data $u_0 \in \fra H$, we denote the solution at time $t$ by 
\begin{equation}
\label{S_t introduction}
u(t) \;=:\; S_t u_0\,.
\end{equation}

\subsection{The quantum problem}
We use the same conventions as in \cite[Section 1.4]{FrKnScSo}.
We work on the \emph{bosonic Fock space}
\begin{equation*}
\cal F \;\equiv\; \cal F(\fra H) \:\deq\: \bigoplus_{p \in \N} \fra H^{(p)}\,.
\end{equation*} 
Here, for $p \in \N$, the \emph{$p$-particle space} $\fra H^{(p)}$ is defined as the symmetric subspace of $\fra H^{\otimes p}$.
For $f \in \fra H$ let $b^*(f)$ and $b(f)$ denote the usual bosonic creation and annihilation operators on $\cal F$, defined by
\begin{align} \label{def_b1}
\pb{b^*(f) \Psi}^{(p)}(x_1, \dots, x_p) &\;=\; \frac{1}{\sqrt{p}} \sum_{i = 1}^p f(x_i) \Psi^{(p - 1)}(x_1, \dots, x_{i - 1}, x_{i+1}, \dots, x_p)\,,
\\ \label{def_b2}
\pb{b(f) \Psi}^{(p)}(x_1, \dots, x_p) &\;=\; \sqrt{p+1} \int \dd x \, \bar f(x) \, \Psi^{(p+1)} (x,x_1, \dots, x_p)\,,
\end{align}
where we denote vectors of $\cal F$ by $\Psi = (\Psi^{(p)})_{p \in \N}$. They satisfy the canonical commutation relations 
\begin{equation*}
[b(f),b^*(g)]\;=\;\langle f, g \rangle_{\fra H}\,,\quad \,[b(f),b(g)]\;=\;[b^*(f),b^*(g)]\;=\;0\,.
\end{equation*}
We define the rescaled creation and annihilation operators $\phi_\tau^*(f) \deq \tau^{-1/2} \, b^*(f)$ and $\phi_\tau(f) \deq \tau^{-1/2} \, b(f)$. We think of $\phi_\tau^*$ and $\phi_\tau$ as operator-valued distributions and we denote their distribution kernels as $\phi_\tau^*(x)$ and $\phi_\tau(x)$ respectively. In analogy to the classical field $\phi$ defined in \eqref{classical_free_field} below, we call $\phi_\tau$ the quantum field. For more details, we refer the reader to \cite[Section 1.4]{FrKnScSo}.

Let $p \in \N$ and $\xi$ a closed linear operator on $\fra H^{(p)}$, given by a Schwartz integral kernel that we denote by $\xi(x_1, \dots, x_p; y_1, \dots, y_p)$; see \cite[Corollary V.4.4]{RS1}. We define the 
\emph{lift} of $\xi$ to $\cal F$ by
\begin{equation}
\label{Theta_tau_xi}
\Theta_\tau(\xi) \;\deq\; \int \dd x_1 \cdots \dd x_p \, \dd y_1 \cdots \dd y_p \, \xi(x_1, \dots, x_p; y_1, \dots, y_p) \, \phi_\tau^*(x_1) \cdots \phi_\tau^*(x_p) \phi_\tau(y_1) \cdots \phi_\tau(y_p)\,.
\end{equation}
The \emph{quantum interaction} is defined as
\begin{equation}
\label{W_tau}
\cal W_\tau \;\deq\; \frac{1}{2}\,\Theta_\tau(W) \;=\; \frac{1}{2} \int \dd x \; \dd y\; \phi_\tau^*(x) \phi_\tau^*(y) \,w(x-y)\; \phi_\tau(x) \phi_\tau(y)\,.
\end{equation} 
Here $W \equiv W^{(2)}$ is the two particle operator on $\fra H^{(2)}$ given by multiplication by $w(x_1-x_2)$ for $w$ as in \eqref{definition_w}. 
The \emph{free quantum Hamiltonian} is given by
\begin{equation}
\label{H_tau_0}
H_{\tau,0}\;\deq\; \Theta_\tau(h) \;=\; \int \dd x \, \dd y \,  \phi^*_\tau(x) \,h(x;y) \,  \phi_\tau(y)\,.
\end{equation}
The \emph{interacting quantum Hamiltonian} is defined as
\begin{equation}
\label{H_tau}
H_\tau\;\deq\;H_{\tau,0}+\cal W_\tau\,.
\end{equation} 
The \emph{grand canonical ensemble} is defined as $P_\tau \deq \ee^{-H_\tau}$.
We define the \emph{quantum state} $\rho_\tau(\cdot)$ as
\begin{equation}
\label{rho_tau}
\rho_\tau(\cal A) \;\deq\; \frac{\tr (\cal A P_\tau)}{\tr (P_\tau)}\,
\end{equation}
for $\cal A$ a closed operator on $\cal F$. In what follows, it is helpful to work with the \emph{rescaled} version of the interacting quantum Hamiltonian given by $\tau H_\tau$. 

\begin{definition}
\label{Quantum_time_evolution}
Let $\mathbf{A}$ be an operator on the Fock space $\cal F$. 
We define its quantum time evolution as
\begin{equation*}
\Psi_{\tau}^t \, \mathbf{A} \;\deq\;\ee^{\ii t \tau H_\tau} \mathbf{A}\, \ee^{-\ii t \tau H_\tau}\,.
\end{equation*}
\end{definition}

\subsection{The classical problem}

For each $k \in \N$, let $\mu_k$ be a standard complex Gaussian measure, i.e.\ $\mu_k (\dd z) = \frac{1}{\pi} \ee^{-|z|^2} \dd z$, where $\dd z$ is the Lebesgue measure on $\C$. We then introduce the probability space $(\C^\N, \cal G, \mu)$, with $\cal G$ the product sigma-algebra and the product probability measure 
\begin{equation}
\label{measure_mu}
\mu\;\deq\;\bigotimes_{k \in \N} \mu_k. 
\end{equation}
Elements of the corresponding probability space $\C^{\N}$ are denoted by $\omega=(\omega_k)_{k \in \N}$.

We denote by $\phi \equiv \phi(\omega)$ the \emph{free classical field}
\begin{equation}
\label{classical_free_field} 
\phi \;\deq\; \sum_{k \in \N} \frac{\omega_k}{\sqrt{\lambda_k}} \, u_k\,.
\end{equation}
Note that, by \eqref{Tr h^{s-1}}, the sum \eqref{classical_free_field} converges in $\fra H$ almost surely.

For a closed operator $\xi$ on $\fra H^{(p)}$, in analogy to \eqref{Theta_tau_xi}, we define the random variable
\begin{equation}
\label{theta_xi}
\Theta(\xi) \;\deq\; \int \dd x_1 \cdots \dd x_p \, \dd y_1 \cdots \dd y_p \, \xi(x_1, \dots, x_p; y_1, \dots, y_p) \, \bar \phi(x_1) \cdots \bar \phi(x_p) \phi(y_1) \cdots \phi(y_p)\,.
\end{equation}
Note that if $\xi$ is a bounded operator then $\Theta(\xi)$ is almost surely well-defined, since $\phi \in \fra H$ almost surely.

Given $w$ as in \eqref{definition_w}, the \emph{classical interaction} is defined as
\begin{equation} \label{classical interaction}
\cal W \;\deq\; \frac{1}{2} \Theta(W) \;=\; \frac{1}{2} \int \dd x \, \dd y \, \abs{\phi(x)}^2 \, w(x - y) \, \abs{\phi(y)}^2\,.
\end{equation}
Moreover, the \emph{free classical Hamiltonian} is given by
\begin{equation}
\label{classical_free_hamiltonian}
H_0 \;\deq\; \Theta(h) \;=\; \int \dd x \, \dd y \,  \bar{\phi}(x) \,h(x;y) \,  \phi(y)\,.
\end{equation}
The
\emph{interacting classical Hamiltonian} is given by
\begin{equation}
\label{classical_hamiltonian}
H \;\deq\: H_0 + \cal W\,.
\end{equation}

We define the \emph{classical state} $\rho(\cdot)$ as
\begin{equation} \label{rho_frac}
\rho(X) \;\deq\; \frac{\int X \,\ee^{-\cal W} \,\dd \mu}{\int \ee^{-\cal W}\,\dd \mu}\,,
\end{equation}
where $X$ is a random variable.

\begin{definition}
\label{Classical_time_evolution}
Let $p \in \N$ and $\xi$ be a bounded operator on $\fra H^{(p)}$.
We define the random variable
\begin{equation*}
\Psi^t \, \Theta(\xi) \;\deq\; \int \dd x_1 \cdots \dd x_p \, \dd y_1 \cdots \dd y_p \, \xi(x_1, \dots, x_p; y_1, \dots, y_p) \, \overline{S_t \phi}(x_1) \cdots \overline{S_t\phi}(x_p) \,S_t \phi(y_1) \cdots S_t \phi(y_p)\,,
\end{equation*}
where $S_t$  is the flow map from \eqref{S_t introduction}. Note that $\Psi^t \, \Theta(\xi)$ is well defined since $\phi \in \fra H$ almost surely and since $S_t$ preserves the norm on $\fra H$.
\end{definition}

\subsection{Statement of the main results}
We denote by $\cal L(\cal H)$ the space of bounded operators on a Hilbert space $\cal H$. We prove the following result for the flow of \eqref{NLS1}.
\begin{theorem}[Convergence of time-dependent correlation functions for the nonlocal nonlinearity]
\label{Main Result}
Given $m \in \N$, $p_1, \ldots, p_m \in \N$, $\xi^1 \in \cal L(\fra H^{(p_1)}), \ldots, \xi^m \in \cal L(\fra H^{(p_m)})$ and $t_1,\ldots,t_m \in \R$, we have
\begin{equation*}
\lim_{\tau \to \infty} \rho_\tau \Big(\Psi_\tau^{t_1} \,\Theta_\tau(\xi^1)\,\cdots \, 
\Psi_\tau^{t_m} \,\Theta_\tau(\xi^m)\Big) \;=\; \rho \Big(\Psi^{t_1} \,\Theta(\xi^1)\,\cdots \, 
\Psi^{t_m} \,\Theta(\xi^m)\Big)\,.
\end{equation*}
\end{theorem}
\begin{remark}
\label{classical measure remark}
For all $p \in \N$, $\xi \in \cal L(\fra H^{(p)})$, and $t \in \R$ we have by \eqref{rho_tau}, Definition \ref{Quantum_time_evolution}, and the cyclicity of the trace that
\begin{equation}
\label{cyclicity_of_the_trace}
\rho_\tau \big(\Psi_\tau^t\,\Theta_\tau (\xi)\big) \;=\; \rho_\tau \big(\Theta_\tau(\xi)\big)
\end{equation}
for all $\tau$. In particular, substituting \eqref{cyclicity_of_the_trace} into Theorem \ref{Main Result} with $m=1$, it follows that
\begin{equation}
\label{Equality_of_moments1}
\rho \big(\Psi^t\,\Theta(\xi)\big) \;=\; \rho \big(\Theta(\xi)\big)\,.
\end{equation}
Hence, using \eqref{cyclicity_of_the_trace}--\eqref{Equality_of_moments1}, we recover the invariance of the Gibbs measure for \eqref{NLS1}, proved in \cite{B}.
\end{remark}

Choosing physical space to be a circle, $\Lambda=\mathbb{T}^1$, and the external potential to vanish, $v=0$, we prove an analogue of Theorem \ref{Main Result} for the dynamics corresponding to a local nonlinearity (see \eqref{localNLS1}) by using an approximation argument. Let $w$ be a continuous compactly supported nonnegative function satisfying $\int \dd x\, w(x)=1$. For $\epsilon>0$ we define the two-body potential
\begin{equation}
\label{w_epsilon}
w^\epsilon(x) \;\deq\; \frac{1}{\epsilon} \,w\bigg(\frac{[x]}{\epsilon}\bigg)\,.
\end{equation}
Here, and in the sequel, $[x]$ denotes the unique element of the set $(x + \Z) \cap [-1/2,1/2)$.

\begin{theorem}[Convergence of time-dependent correlation functions for a local nonlinearity]
\label{Main Result_local_nonlinearity}
Suppose that $\Lambda=\mathbb{T}^1$, $v=0$, and $w^{\varepsilon}$ is defined as in \eqref{w_epsilon}. There exists a sequence $(\epsilon_\tau)$ of positive numbers satisfying $\lim_{\tau \to \infty} \epsilon_\tau = 0$, such that, for arbitrary $m \in \N$, $p_1, \ldots, p_m \in \N$, $\xi^1 \in \cal L(\fra H^{(p_1)}), \ldots, \xi^m \in \cal L(\fra H^{(p_m)})$, and $t_1 \in \mathbb{R},\ldots,t_m \in \R$, we have
\begin{equation*}
\lim_{\tau \to \infty} \rho_\tau^{\epsilon_\tau} \Big(\Psi_\tau^{t_1,\epsilon_\tau} \,\Theta_\tau(\xi^1)\,\cdots \, 
\Psi_\tau^{t_m,\epsilon_\tau} \,\Theta_\tau(\xi^m)\Big) \;=\; \rho \Big(\Psi^{t_1} \,\Theta(\xi^1)\,\cdots \, 
\Psi^{t_m} \,\Theta(\xi^m)\Big)\,.
\end{equation*}
Here, the quantum state $\rho^{\epsilon}_\tau(\cdot)$ is defined in \eqref{rho_tau} and the quantum-mechanical time evolution $\Psi_\tau^{t,\epsilon}$ is introduced in Definition \ref{Quantum_time_evolution}, where the two-body potential is $w^{\epsilon}$. Moreover, the classical state $\rho(\cdot)$ is defined in \eqref{rho_frac} and the classical time evolution $\Psi^t$ is introduced in Definition \ref{Classical_time_evolution}, where the two-body potential is $w=\delta$. (Hence the classical time evolution is governed by the local nonlinear Schr\"{o}dinger equation \eqref{localNLS1}).
\end{theorem}

\begin{remark}
As in Remark \ref{classical measure remark}, Theorem \ref{Main Result_local_nonlinearity} allows us to establish the invariance of the Gibbs measure for \eqref{localNLS1} first proved in \cite{B}.
\end{remark}

\begin{remark}
For interacting Bose gases on a finite lattice, results similar to Theorems \ref{Main Result} and \ref{Main Result_local_nonlinearity} have been obtained in \cite[Section 3.4]{Knowles_Thesis}.
\end{remark}

\subsubsection*{Conventions}
We denote by $C$ a positive constant that can depend on the fixed quantities of the problem (for example the interaction potential $w$). This constant can change from line to line. If it depends on a family of parameters $a_1,a_2,\ldots$, we write $C=C(a_1,a_2,\ldots)$. Given a separable Hilbert space $\cal H$ and $q \in [1,\infty]$, we denote by $\fra S^q(\cal H)$ the $q$-\emph{Schatten class}. This is the set of all $\cal T \in \cal L(\cal H)$ 
such that the norm given by
\begin{equation*}
\|\cal T\|_{\fra S^q(\cal H)} \;\deq\;
\begin{cases}
( \tr \, \abs{\cal T}^q)^{1/q}  &\mbox{if }q<\infty\\
\sup \spec \, \abs{\cal T} &\mbox{if } q=\infty\,
\end{cases}
\end{equation*}
is finite.
Here we recall that $|\cal T| \deq \sqrt{\cal T^* \cal T}$.
In particular, we note that by definition $\fra S^\infty(\cal H)=\cal L(\cal H)$, the space of bounded operators on $\cal H$. We abbreviate the operator norm $\|\cdot\|_{\fra S^{\infty}}$ by $\|\cdot\|$.
Any quantity bearing a subscript $\tau$ is a quantum object and any quantity not bearing this subscript is a classical object.

\section{Strategy of the proof}

We first outline the strategy of proof of Theorem \ref{Main Result}, concerning the nonlocal problem. Let us recall several definitions.
The \emph{rescaled number of particles} is defined as
\begin{equation}
\label{N_tau}
\cal N_\tau \;\deq\; \int \dd x \, \phi_\tau^*(x)\,\phi_\tau(x)\,.
\end{equation}
Moreover, the \emph{mass} is defined as
\begin{equation}
\label{definition_mass}
\cal N \;\deq\; \int \dd x \, |\phi(x)|^2\,.
\end{equation}
Theorem \ref{Main Result} can be deduced from the following two propositions.

\begin{proposition}[Convergence in the small particle number regime]
\label{Small_N_convergence}
Let $F \in C_c^{\infty}(\R)$ with $F \geq 0$ be given.
Given $m \in \N$, $p_1, \ldots, p_m \in \N$, $\xi^1 \in \cal L(\fra H^{(p_1)}), \ldots, \xi^m \in \cal L(\fra H^{(p_m)})$, and $t_1,\ldots,t_m \in \R$, we have
\begin{equation*}
\lim_{\tau \to \infty} \rho_\tau \Big(\Psi_\tau^{t_1} \,\Theta_\tau(\xi^1)\,\cdots \, 
\Psi_\tau^{t_m} \,\Theta_\tau(\xi^m)\,F(\cal N_\tau)
\Big) \;=\; \rho \Big(\Psi^{t_1} \,\Theta(\xi^1)\,\cdots \, 
\Psi^{t_m} \,\Theta(\xi^m)\,F\big(\cal N)
\Big)\,.
\end{equation*}
\end{proposition}

\begin{proposition}[Bounds in the large particle number regime]
\label{Large_N_convergence}
Let $G \in C^{\infty}(\R)$ be such that $0 \leq G \leq 1$ and $G=0$ on $[0,\cal K]$ for some $\cal K>0$.
Furthermore, let $m \in \N$, $p_1, \ldots, p_m \in \N$, $\xi^1 \in \cal L(\fra H^{(p_1)}), \ldots, \xi^m \in \cal L(\fra H^{(p_m)})$, and $t_1,\ldots,t_m \in \R$ be given. The following estimates hold.
\begin{itemize}
\item[(i)] $\Big|\rho_\tau \Big(\Psi_\tau^{t_1} \,\Theta_\tau(\xi^1)\,\cdots \, 
\Psi_\tau^{t_m} \,\Theta_\tau(\xi^m)\,G(\cal N_\tau)
\Big)\Big| \leq \frac{C}{\cal K}.$
\item[(ii)] $ \Big|\rho \Big(\Psi^{t_1} \,\Theta (\xi^1)\,\cdots \, 
\Psi^{t_m} \,\Theta(\xi^m)\,G(\cal N)
\Big)\Big| \leq \frac{C}{\cal K}.$
\end{itemize}
Here $C=C(\|\xi^1\|,\ldots,\|\xi^m\|,p_1+\cdots+p_m)>0$ is a constant that does not depend on $\cal K$.
\end{proposition}

\begin{proof}[Proof of Theorem \ref{Main Result}] For fixed $\cal K>0$, we choose $F \equiv F_{\cal K}$ in Proposition \ref{Small_N_convergence} such that $0 \leq F \leq 1$ and $F=1$ on $[0,\cal K]$ and we let $G \equiv G_{\cal K} \deq 1-F_{\cal K}$ in Proposition \ref{Large_N_convergence}. We then deduce Theorem \ref{Main Result} by letting $\cal K \rightarrow \infty$. 
\end{proof}

We prove Proposition \ref{Small_N_convergence} in Section \ref{Small particle number} and Proposition \ref{Large_N_convergence} in Section \ref{Large particle number} below. 

Theorem \ref{Main Result_local_nonlinearity}, concerning the local problem, is proved in Section \ref{The local problem} by using Theorem \ref{Main Result} and a limiting argument. At this step, it is important to prove an $L^2$-convergence result of solutions of the NLS with interaction potential given by \eqref{w_epsilon} to solutions of \eqref{localNLS1}. The precise statement is given in Proposition \ref{L2 convergence} below. Note that, in order to prove this statement, it is not enough to use energy methods, but we have to directly use the dispersion in the problem. To this end, we use $X^{s,b}$ spaces, which are recalled in Definition \ref{X^{sigma,b}} below.

\section{The small particle number regime: proof of Proposition \ref{Small_N_convergence}.}
\label{Small particle number}

In this section we consider the small particle number regime and prove Proposition \ref{Small_N_convergence}.

In what follows, it is useful to note that, given $\xi \in \cal L (\fra H^{(p)})$, for $\Theta_\tau(\xi)$ defined as in \eqref{Theta_tau_xi}, we have 
\begin{equation}
\label{n_sector}
\Theta_\tau(\xi)\big|_{\fra H^{(n)}}
\;=\;
\begin{cases}
\frac{p!}{\tau^p} {n \choose p} P_{+}\big(\xi \otimes \mathbf{1}^{(n-p)}\big)P_{+} &\mbox{if }n \geq p
\\
0 & \mbox{otherwise. }
\end{cases}
\end{equation}
(For more details see \cite[(3.88)]{Knowles_Thesis}.)
Here $\mathbf{1}^{(q)}$ denotes the identity map on $\fra H^{(q)}$ and $P_{+}$ denotes the orthogonal projection onto the subspace of symmetric tensors.
In particular (c.f.\ \cite[Section 3.4.1]{Knowles_Thesis}), we deduce the following estimate.
\begin{lemma}
\label{Quantum Lemma 2}
Let $\xi \in \cal L(\fra H^{(p)})$ be given. For all $n \in \N$ we have
\begin{equation*}
\Big\|\Theta_\tau(\xi)\big|_{\fra H^{(n)}}\Big\| \;\leq\;\Big(\frac{n}{\tau}\Big)^p \, \|\xi\|\,.
\end{equation*}
\end{lemma}
Moreover, by applying the Cauchy-Schwarz inequality, we obtain the following result in the classical setting (c.f.\ also [Section 3.4.2]\cite{Knowles_Thesis}).
\begin{lemma}
\label{Classical Lemma 2}
Let $\xi \in \cal L(\fra H^{(p)})$ be given. Then we have 
\begin{equation*}
|\Theta(\xi)| \;\leq\; \|\phi\|_{\fra H}^{2p}\, \|\xi\|\,.
\end{equation*}
\end{lemma}

\subsection{An auxiliary convergence result}
In the proof of Proposition \ref{Small_N_convergence}, we use the following auxiliary convergence result.

For $p \in \N$ define the unit ball $\fra B_p \deq \{\eta \in \fra S^2(\fra H^{(p)}):\|\eta\|_{\fra S^2(\fra H^{(p)})}\leq 1\}$.
\begin{proposition}
\label{Convergence 1}
Let $f \in C_c^\infty(\R)$ be given.
\begin{itemize}
\item[(i)]
We have 
\begin{equation}
\label{eq:rho theta f convergence}
\lim_{\tau \to \infty} \rho_\tau\big(\Theta_\tau(\xi) f(\cal N_\tau)\big) \;=\; \rho\big(\Theta(\xi) f(\cal N)\big)\,,
\end{equation}
uniformly in $\xi \in \fra B_p \cup \{\mathbf{1}^{(p)}\}$.
\item[(ii)] Moreover, if $f \geq 0$ then  \eqref{eq:rho theta f convergence} holds for all $\xi \in \cal L(\fra H^{(p)})$.
\end{itemize}
\end{proposition}

We note that, if $f$ were equal to $1$ (which is not allowed in the assumptions), then the result of Proposition \ref{Convergence 1} would follow immediately from \cite[Theorem 5.3]{Lewin_Nam_Rougerie} or equivalently \cite[Theorem 1.8]{FrKnScSo}.
However, Proposition \ref{Convergence 1} does not immediately follow from the arguments in \cite{FrKnScSo} since the presence of $f$ breaks the Gaussian structure which allows us to apply the Wick theorem. In the proof we expand $f$ by means of complex analytic methods in such a way that we can apply the analysis of \cite{FrKnScSo} in the result.

Before we proceed with the proof of Proposition \ref{Convergence 1}, we introduce some notation and collect several auxiliary results.

For $\cal N$ as in \eqref{definition_mass} and $\nu>0$ we define the measure 
\begin{equation*}
\dd \tilde \mu^\nu \;\deq\; \ee^{-\nu \cal N} \dd \mu\,.
\end{equation*}
Note that $\dd \tilde \mu^\nu$ is still Gaussian, but it is not normalized. Indeed, recalling \eqref{measure_mu}, we write
\begin{equation*}
\dd \mu(\omega) \;=\; \bigotimes_{k \in \N} \frac{1}{\pi} \ee^{-\abs{\omega_k}^2} \dd \omega_k\,.
\end{equation*}
We have
\begin{equation*}
\cal N \;=\; \int \dd x \, \abs{\phi(x)}^2 \;=\; \sum_{k \in \N} \frac{\abs{\omega_k}^2}{\lambda_k}\,,
\end{equation*}
and so we find
\begin{equation}
\label{d_tilde_mu}
\dd \tilde \mu^\nu (\omega) \;=\; \bigotimes_{k \in \N} \frac{1}{\pi} \ee^{- \nu \abs{\omega_k}^2 / \lambda_k} \ee^{-\abs{\omega_k}^2} \dd \omega_k\,.
\end{equation}
Define the normalized Gaussian measure
\begin{equation*}
\dd \mu^\nu \;\deq\; \frac{\dd \tilde \mu^\nu}{\int \dd \tilde \mu^\nu}\,.
\end{equation*}
The measure $\mu^\nu$ satisfies a Wick theorem, where any moment of variables that are linear functions of $\phi$ or $\bar{\phi}$ is given as a sum over pairings and each pair is computed using the (Hermitian) covariance of $\mu^\nu$ given by
\begin{equation}
\label{h_nu}
h^\nu \;\deq\; h+\nu \;=\; \sum_{k \in \N} (\lambda_k+\nu) u_k u_k^*\,.
\end{equation} 
In terms of $\phi$ we have
\begin{equation*}
\int \dd \mu^\nu \bar \phi(g) \, \phi(f) \;=\; \scalar{f}{(h^\nu)^{-1} g}\,.
\end{equation*}
In the above identity, we write $\phi(f) \deq \scalar{f}{\phi}$ and $\bar \phi (g) \deq \scalar{\phi}{g}$.
For $\re z \geq 0$ and for $X$ a random variable and for $\cal W$ as in \eqref{classical interaction}, we define the deformed classical state
\begin{equation}
\label{tilde_rho_definition}
\tilde \rho_z^\nu(X) \;\deq\; \int X \, \ee^{-z \cal W} \, \dd \mu^\nu\,.
\end{equation}

In the quantum setting, for $\re z \geq 0$ and for $\cal A$ a closed operator on $\cal F$, we define the deformed quantum state
\begin{equation}
\label{tilde_rho_tau_definition}
\tilde \rho_{\tau,z}^\nu(\cal A) \;\deq\; \frac{\tr \pb{\cal A \, \ee^{-H_{\tau,0} - z \cal W_\tau- \nu \cal N_\tau}}}{\tr (\ee^{-H_{\tau,0} - \nu \cal N_\tau})}\,.
\end{equation}
The free state $\tilde \rho_{\tau,0}^\nu(\cdot)$ satisfies a quantum Wick theorem (c.f.\ \cite[Appendix B]{FrKnScSo}), with the quantum Green function
\begin{equation*}
G_\tau^\nu \;\deq\; \frac{1}{\tau (\ee^{h^\nu/\tau} - 1)}\,.
\end{equation*}
In the proof of Proposition \ref{Convergence 1} we have to analyse
\begin{equation*}
\frac{\tr \pb{\cal A \, \ee^{-H_{\tau,0} - z \cal W_\tau- \nu \cal N_\tau}}}{\tr (\ee^{-H_{\tau,0}})} \;=\; \tilde \rho_{\tau,z}^\nu(\cal A) \, \frac{\tr (\ee^{-H_{\tau,0} - \nu \cal N_\tau})}{\tr (\ee^{-H_{\tau,0}})}\,.
\end{equation*}
With the above notation we have the following result.
\begin{lemma}
\label{Convergence 2}
For $\nu>0$ we have
\begin{equation*}
\lim_{\tau \to \infty}\frac{\tr \pb{\ee^{-H_{\tau,0}- \nu \cal N_\tau}}}{\tr (\ee^{-H_{\tau,0}})} \;=\; \int \dd \tilde \mu^\nu\,.
\end{equation*}
\end{lemma}
\begin{proof}[Proof of Lemma \ref{Convergence 2}]
A direct calculation using \eqref{d_tilde_mu} shows that
\begin{equation}
\label{eq:Convergence 2 (ii) RHS}
\int \dd \tilde \mu^\nu \;=\;\prod_{k \in \N} \frac{\lambda_k}{\lambda_k+\nu}\,.
\end{equation}
By using the occupation state basis (c.f.\ \cite[Appendix B, Proof of Lemma B.1]{FrKnScSo}), we have
\begin{equation}
\label{eq:Convergence 2 (ii) LHS}
\frac{\tr \pb{\ee^{-H_{\tau,0}- \nu \cal N_\tau}}}{\tr (\ee^{-H_{\tau,0}})}  \;=\;
\frac{\sum_{\vec m}\ee^{-\sum_k \frac{\lambda_k+\nu}{\tau}m_k}}{\sum_{\vec m}\ee^{-\sum_k \frac{\lambda_k}{\tau}m_k}}\;=\;\frac{\sum_{\vec m}\prod_k \ee^{-\frac{\lambda_k+\nu}{\tau}m_k}}{\sum_{\vec m}\prod_k\ee^{-\frac{\lambda_k}{\tau}m_k}}\;=\;\prod_{k \in \N} \frac{1-\ee^{-\frac{\lambda_k}{\tau}}}{1-\ee^{-\frac{\lambda_k+\nu}{\tau}}}\,.
\end{equation}

We note that, for fixed $k \in \N$, we have
\begin{equation}
\label{eq:Convergence 2 term bound}
\frac{1-\ee^{-\frac{\lambda_k}{\tau}}}{1-\ee^{-\frac{\lambda_k+\nu}{\tau}}}\;=\;1+\frac{\ee^{-\frac{\lambda_k}{\tau}}(\ee^{-\frac{\nu}{\tau}}-1)}{1-\ee^{-\frac{\lambda_k+\nu}{\tau}}}\;=\;1+\cal O\Big(\frac{\nu}{\lambda_k}\Big)\,.
\end{equation}
Indeed, we note that by the mean value theorem 
\begin{equation}
\label{eq:Convergence 2 term bound 1a}
\ee^{-\frac{\nu}{\tau}}-1\;=\;\cal O\Big(\frac{\nu}{\tau}\Big)\,.
\end{equation}
Also, we observe that 
\begin{equation}
\label{eq:Convergence 2 term bound 1b}
\Big|1-\ee^{-\frac{\lambda_k+\nu}{\tau}}\Big| \;\geq\; \Big|1-\ee^{-\frac{\lambda_k}{\tau}}\Big|\,.
\end{equation}
If $\lambda_k < \tau$ we have
\begin{equation}
\label{eq:Convergence 2 term bound a}
\Big|1-\ee^{-\frac{\lambda_k}{\tau}}\Big| \;\geq \; \frac{C \lambda_k}{\tau}\,, \quad \ee^{-\frac{\lambda_k}{\tau}} \;\leq\; 1\,.
\end{equation}
Furthermore, if $\lambda_k \geq \tau$ we have
\begin{equation}
\label{eq:Convergence 2 term bound b}
\Big|1-\ee^{-\frac{\lambda_k}{\tau}}\Big|  \;\geq\; C\,,\quad \ee^{-\frac{\lambda_k}{\tau}} \;\leq\;  \frac{C \tau}{\lambda_k}\,.
\end{equation}
The estimate \eqref{eq:Convergence 2 term bound} follows from \eqref{eq:Convergence 2 term bound 1a}--\eqref{eq:Convergence 2 term bound b}.

Note that the individual factors of \eqref{eq:Convergence 2 (ii) LHS} converge to the corresponding factors of \eqref{eq:Convergence 2 (ii) RHS} as $\tau \rightarrow \infty$. 
We hence reduce the claim to showing that
\begin{equation}
\label{eq:Convergence 2 product}
\lim_{\tau \to \infty} \mathop{\prod_{k \in \N:}}_{\lambda_k \gg \nu} \frac{1-\ee^{-\frac{\lambda_k}{\tau}}}{1-\ee^{-\frac{\lambda_k+\nu}{\tau}}} \;=\; \mathop{\prod_{k \in \N:}}_{\lambda_k \gg \nu} \frac{\lambda_k}{\lambda_k+\nu}\,.
\end{equation}
Here we use the notation $A \gg B$ if there exists a large constant $C>0$ such that $A \geq C B$. (The size of $C$ is specified from context).
The convergence \eqref{eq:Convergence 2 product} follows from the dominated convergence theorem after taking logarithms on both sides, using \eqref{eq:Convergence 2 term bound}, the inequality $|\log(1+z)| \leq C|z|$ for $|z| \leq 1/2$ and the assumption that $\tr h^{-1}<\infty$. Taking logarithms is justified by \eqref{eq:Convergence 2 term bound} and the assumption $\lambda_k \gg \nu$.
\end{proof}

We now have all of the necessary ingredients to prove Proposition \ref{Convergence 1}.

\begin{proof}[Proof of Proposition \ref{Convergence 1}]
We first prove (i). 
Let us consider 
\begin{equation*}
\xi \in \cal C_p \;\deq\; \cal B_p \cup \{\mathbf{1}^{(p)}\}\,.
\end{equation*}
For $\zeta \in \mathbb{C} \setminus [0,\infty)$, we define the functions $\alpha^\xi_\tau \equiv \alpha^\xi_\tau(\zeta)$ and $\alpha^\xi \equiv \alpha^\xi(\zeta)$ by
\begin{equation}
\label{eq:Function alphaI}
\alpha^\xi_{\sharp}(\zeta) \;\deq\; \rho_\sharp\bigg(\Theta_\sharp(\xi)\,\frac{1}{\cal N_\sharp-\zeta}\bigg)\,.
\end{equation}
In the above formula and throughout the proof of the proposition, we use the convention that, given $Y = \cal N, \alpha ,\rho, \ldots$,  the quantity $Y_\sharp$ formally
denotes either $Y_\tau$ or $Y$. In this convention, we write $\phi^*$ for $\bar{\phi}$. This simplifies some of the notation in the sequel.

For $\re \zeta <0$ we have
\begin{equation}
\label{eq:Function alphaII}
\frac{1}{\cal N_\sharp-\zeta} \;=\; \int_{0}^{\infty} \dd \nu\,\ee^{-(\cal N_\sharp-\zeta)\nu}\;=\; \int_{0}^{\infty} \dd \nu\, \ee^{\zeta \nu}\,\ee^{-\nu \cal N_\sharp}\,.
\end{equation}
In particular, from \eqref{eq:Function alphaI}--\eqref{eq:Function alphaII}, it follows that for $\re \zeta <0$ we have 
\begin{equation}
\label{eq:Function alpha 2}
\alpha^\xi_\sharp(\zeta) \;=\;\int_0^{\infty} \dd \nu\, \ee^{\zeta \nu}\,\rho_\sharp \big(\Theta_\sharp (\xi) \ee^{-\nu \cal N_\sharp}\big)\,.
\end{equation}
By Lemma \ref{Quantum Lemma 2} we know that $\pm \Theta_\tau(\xi) \leq \|\xi\| \,\cal N_\tau^p \leq \cal N_\tau^p$ acting on sectors of Fock space (c.f.\ \cite[(3.91)]{Knowles_Thesis}). Hence, it follows that
\begin{equation}
\label{eq:Function alpha I4}
\big|\rho_\tau\big(\Theta_\tau(\xi) \,\ee^{-\nu \cal N_\tau}\big)\big| \;\leq\; \rho_\tau\big(\cal N_\tau^p\,\ee^{-\nu \cal N_\tau}\big) \;\leq\; \rho_\tau\big(\cal N_\tau^p\big)\;\leq C(p)\,,
\end{equation}
uniformly in $\xi \in \cal C_p$ and $\nu>0$. Furthermore, by using Lemma \ref{Classical Lemma 2}, we deduce the classical analogue of  
\eqref{eq:Function alpha I4},
\begin{equation}
\label{eq:Function alpha I4 classical}
\big|\rho\big(\Theta(\xi) \,\ee^{-\nu \cal N}\big)\big| \;\leq\; \rho\big(\cal N^p\,\ee^{-\nu \cal N}\big) \;\leq\; \rho\big(\cal N^p\big)\;\leq C(p)\,,
\end{equation}
uniformly in $\nu>0$. The estimates \eqref{eq:Function alpha I4}--\eqref{eq:Function alpha I4 classical} and the assumption $\re \zeta <0$ allow us to use Fubini's theorem in order to exchange the integration in $\nu$ and expectation $\rho_\sharp(\cdot)$ in \eqref{eq:Function alphaI}--\eqref{eq:Function alphaII} and deduce \eqref{eq:Function alpha 2}.
For fixed $\nu>0$ we have 
\begin{multline}
\label{eq:Function alpha I1}
\rho_\tau\big(\Theta_\tau(\xi) \ee^{-\nu \cal N_\tau}\big) \;=\; \frac{\tr \big(\Theta_\tau(\xi) \,\ee^{-H_{\tau,0}- \cal W_\tau-\nu \cal N_\tau}\big)}{\tr \big(\ee^{-H_{\tau,0}-\nu \cal N_\tau}\big)} \,\frac{\tr \big(\ee^{-H_{\tau,0}-\nu \cal N_\tau}\big)}{\tr \big(\ee^{-H_{\tau,0}}\big)} \, \frac{\tr \big(\ee^{-H_{\tau,0}}\big)}{\tr \big(\ee^{-H_{\tau}}\big)}
\\
\;=\;\tilde \rho_{\tau,1}^\nu\big(\Theta_\tau(\xi)\big) \,\frac{\tr \big(\ee^{-H_{\tau,0}-\nu \cal N_\tau}\big)}{\tr \big(\ee^{-H_{\tau,0}}\big)} \, \frac{1}{\tilde \rho_{\tau,1}^0(1)}\,,
\end{multline}
where $\tilde \rho_{\tau,z}^0$ is defined by setting $\nu=0$ in \eqref{tilde_rho_tau_definition}.
Moreover, we have
\begin{equation}
\label{eq:Function alpha I2}
\rho\big(\Theta(\xi) \ee^{-\nu \cal N}\big) \;=\; \bigg(\int \,\dd \mu^\nu \,\Theta(\xi) \,\ee^{-\cal W}\bigg) \, \bigg(\int \dd \tilde{\mu}^\nu \bigg) \, \bigg(\frac{1}{\int \dd \mu \,\ee^{-\cal W}}\bigg) \;=\;\tilde{\rho}_1^\nu\big(\Theta(\xi)\big) \, \bigg(\int \dd \tilde{\mu}^\nu \bigg) \, \frac{1}{\tilde \rho_{1}^0(1)}\,,
\end{equation}
where $\tilde \rho_{z}^0$ is defined by setting $\nu=0$ in \eqref{tilde_rho_definition}.

We now consider each of the factors in \eqref{eq:Function alpha I1}--\eqref{eq:Function alpha I2}. By \cite[Theorem 1.8]{FrKnScSo} with the Hamiltonian $h^\nu$ given by \eqref{h_nu}, the first factor in \eqref{eq:Function alpha I1} converges to the first factor in \eqref{eq:Function alpha I2} as $\tau \rightarrow \infty$ uniformly in $\xi \in \cal C_p$. Moreover, the convergence of the second factors follows from Lemma \ref{Convergence 2}. Finally, we have convergence of the third factors by \cite[Theorem 1.8]{FrKnScSo} with the Hamiltonian $h$. Note that both applications of \cite[Theorem 1.8]{FrKnScSo} are justified by the assumption \eqref{Tr h^{s-1}}.
In particular, we obtain that
\begin{equation}
\label{eq:Function alpha I3}
\lim_{\tau \to \infty} \rho_\tau\big(\Theta_\tau(\xi) \ee^{-\nu \cal N_\tau}\big) \;=\; \rho\big(\Theta(\xi) \ee^{-\nu \cal N}\big)\,,
\end{equation}
uniformly in $\xi \in \cal C_p$.
From \eqref{eq:Function alpha 2}, \eqref{eq:Function alpha I4}, \eqref{eq:Function alpha I3} and the dominated convergence theorem, it follows that for $\re \zeta<0$ we have
\begin{equation}
\label{eq:Function alpha convergence 1}
\lim_{\tau \to \infty} \alpha_\tau^\xi(\zeta) \;=\; \alpha^\xi(\zeta)\,, 
\end{equation}
uniformly in $\xi \in \cal C_p$.

By \eqref{eq:Function alphaI}, H\"{o}lder's inequality, and arguing as in \eqref{eq:Function alpha I4} and \eqref{eq:Function alpha I4 classical} (with $\nu=0$), we have 
\begin{equation}
\label{eq:Function alpha I5}
|\alpha^\xi_\sharp(\zeta)| \;\leq\; \bigg\| \frac{1}{\cal N_\sharp-\zeta} \bigg\| \,\rho_\sharp \big(\cal N_\sharp^p\big) \;\leq \frac{C(p)}{|\im \zeta|}\,.
\end{equation}
We now show that $\alpha_\tau^\xi,\alpha^\xi$ are analytic in $\mathbb{C} \setminus [0,\infty)$. 

In the sequel, we use the notation 
\begin{equation}
\label{H(R)}
\fra H^{(\leq R)} \;\deq\; \bigoplus_{p \leq R} \fra H^{(p)}\,,\quad \fra H^{(\geq R)} \;\deq\; \bigoplus_{p \geq R} \fra H^{(p)}
\end{equation}
for $R>0$. Let us denote the corresponding orthogonal projections by
\begin{equation}
\label{P(R)}
P^{(\leq R)} : \cal F \rightarrow \fra H^{(\leq R)}\,,\quad P^{(\geq R)} : \cal F \rightarrow \fra H^{(\geq R)}\,.
\end{equation}

In order to prove the analyticity of $\alpha_\tau^\xi$ in $\mathbb{C} \setminus [0,\infty)$ we argue similarly as in the proof of \cite[Lemma 2.34]{FrKnScSo}. Namely, given $n \in \N$ we define for $\zeta \in \mathbb{C} \setminus [0,\infty)$
\begin{equation*}
\alpha^\xi_{\tau,n}(\zeta) \;\deq\; \rho_\tau\bigg(P^{(\leq n)}\,\Theta_\tau(\xi)\,\frac{1}{\cal N_\tau-\zeta}\bigg)\,.
\end{equation*}
Here $P^{(\leq n)}$ is defined as in \eqref{P(R)}.
Note that $\alpha_{\tau,n}^\xi$ is analytic in $\mathbb{C} \setminus [0,\infty)$ since $\cal N_\tau$ is constant on each $m$-particle sector of $\cal F$. As in \eqref{eq:Function alpha I5} we note that 
\begin{equation*}
|\alpha_{\tau,n}^\xi(\zeta)| \;\leq\; C(p) \, \bigg\|\frac{1}{\cal N_\tau-\zeta}\bigg\| \;\leq\; \frac{C(p)}{\max\{-\re \zeta,|\im \zeta|\}}\,.
\end{equation*}
Finally, for $\zeta \in \mathbb{C} \setminus [0,\infty)$ we know that $\lim_{n \rightarrow \infty} \alpha_{\tau,n}^\xi(\zeta)=\alpha_\tau^\xi(\zeta)$ by construction. The analyticity of $\alpha_\tau^\xi(\zeta)$ in $\mathbb{C} \setminus [0,\infty)$ now follows. The analyticity of $\alpha^\xi$ in $\mathbb{C} \setminus [0,\infty)$ is verified by using Lemma \ref{Classical Lemma 2} and differentiating under the integral sign in the representation
\begin{equation*}
\alpha^\xi(\zeta) \;=\; \frac{\int \dd \mu \, \Theta(\xi) \, \frac{1}{\cal N-\zeta}\, \ee^{-W}}{\int \dd \mu \, \ee^{-W}}\,.
\end{equation*}

In what follows, we define the function $\beta_\tau^\xi : \mathbb{C} \setminus [0,\infty) \rightarrow \mathbb{C}$ by
\begin{equation}
\label{eq:Function beta}
\beta_\tau^\xi\;\deq\;\alpha_\tau^\xi-\alpha^\xi\,.
\end{equation}
From the analyticity of $\alpha_\sharp^\xi$ on $\mathbb{C} \setminus [0,\infty)$, \eqref{eq:Function alpha convergence 1} and \eqref{eq:Function alpha I5}, we note that the $\beta_\tau^\xi$ satisfy the following properties.
\begin{enumerate}
\item[(1)] $\beta_\tau^\xi$ is analytic on $\mathbb{C} \setminus [0,\infty)$.
\item[(2)] $\lim_{\tau \to \infty} \sup_{\xi \in \cal C_p} |\beta_\tau^\xi(\zeta)| = 0$ for all $\re \zeta <0$.
\item[(3)] $ \sup_{\xi \in \cal C_p} |\beta_\tau^\xi(\zeta)|\;\leq \frac{C(p)}{|\im \zeta|}$ for all $\zeta \in \mathbb{C} \setminus [0,\infty)$.
\end{enumerate}
We now show that
\begin{equation}
\label{eq:Function beta claim}
\lim_{\tau \rightarrow \infty} \sup_{\xi \in \cal C_p} |\beta_\tau^\xi(\zeta)|=0 \quad \mbox{for all} \quad \zeta \in \mathbb{C} \setminus [0,\infty)\,.
\end{equation}
Namely, we generalise condition (2) above to all $\zeta \in \mathbb{C} \setminus [0,\infty)$.

Given $\epsilon>0$ we define
\begin{equation*}
\cal D_\epsilon \;\deq\; \{\zeta:\im \zeta>\epsilon\}\,
\end{equation*}
and
\begin{equation*}
\cal T_{\epsilon} \;\deq\; \{\zeta_0 \in \cal D_\epsilon: \lim_{\tau \rightarrow \infty}\sup_{\xi \in \cal C_p} \big|\partial_\zeta^m \beta_\tau^\xi (\zeta_0)\big| \rightarrow 0 \,\, \mbox{for all} \,\, m \in \N\}\,.
\end{equation*}
In other words, $\cal T_\epsilon$ consists of all points in $\cal D_\epsilon$ at which all $\zeta$-derivatives of $\beta_\tau^\xi$ converge to zero as $\tau \rightarrow \infty$, uniformly in $\xi \in \cal C_p$.
Note that, by using conditions (1)-(3) above, Cauchy's integral formula and the dominated convergence theorem we have $\cal D_\epsilon \cap \{\zeta: \re \zeta<0\} \subset \cal T_\epsilon$ hence $\cal T_\epsilon \neq \emptyset$.

In order to prove \eqref{eq:Function beta claim} on $\cal D_\epsilon$, it suffices to show that
 $\cal T_\epsilon = \cal D_\epsilon$. By connectedness of $\cal D_\epsilon$ and since $\cal T_\epsilon \neq \emptyset$, the latter claim follows if we prove that $\cal T_\epsilon$ is both open and closed in $\cal D_\epsilon$. Let us first prove that $\cal T_\epsilon$ is open in $\cal D_\epsilon$.
Given $\zeta_0 \in \cal T_\epsilon$, we note that $B_{\zeta_0}(\epsilon/2) \subset \cal D_{\epsilon/2}$. Hence by property (3), it follows that $|\beta_\tau^\xi| \leq C(\epsilon)$ on $B_{\zeta_0}(\epsilon/2)$. By analyticity and Cauchy's integral formula it follows that the series expansion of $\beta_\tau^\xi$ at $\zeta_0$ converges on $B_{\zeta_0}(\epsilon/2)$. Therefore, by differentiating term by term and using the dominated convergence theorem and the assumption $\zeta_0 \in \cal T_\epsilon$, it follows that $B_{\zeta_0}(\epsilon/2) \subset \cal T_\epsilon$. Hence, $\cal T_\epsilon$ is open in $\cal D_\epsilon$.

We now show that $\cal T_\epsilon$ is closed in $\cal D_\epsilon$. Suppose that $(\zeta_n)$ is a 
sequence in $\cal T_\epsilon$ such that $\zeta_n \rightarrow \zeta$ for some $\zeta \in \cal D_\epsilon$. We now show that $\zeta \in \cal T_\epsilon$.
In order to do this, we note that for $n$ large enough we have $\zeta \in B_{\zeta_n}(\epsilon/2)$.
The argument used to show the openness of $\cal T_\epsilon$ in $\cal D_\epsilon$ gives us that $ B_{\zeta_n}(\epsilon/2) \subset \cal T_\epsilon$. In particular, $\zeta \in \cal T_\epsilon$. Hence $\cal T_\epsilon$ is closed in $\cal D_\epsilon$. Therefore $\cal T_\epsilon=\cal D_\epsilon$.
The convergence \eqref{eq:Function beta claim} is shown on $\tilde{\cal D}_\epsilon \;\deq\; \{\zeta:\im \zeta<-\epsilon\}$ by symmetry. The claim \eqref{eq:Function beta claim} on all of $\mathbb{C}
\setminus [0,\infty)$ now follows by letting $\epsilon \rightarrow 0$ and recalling that \eqref{eq:Function beta claim} holds for $\zeta<0$ by condition (2) above.

In what follows we use the notation $\zeta=u+iv$ for $u=\re \zeta$ and $v=\im \zeta$. In particular we have $\partial_{\bar{\zeta}}=\frac{1}{2}(\partial_u+i\partial_v)$.
Applying the Helffer-Sj\"{o}strand formula we obtain 
\begin{equation}
\label{eq:Helffer-Sjostrand_application}
f(\cal N_\sharp) \;=\; \frac{1}{\pi} \int_{\C} \dd \zeta\,\frac{\partial_{\bar{\zeta}}\big[(f(u)+\ii vf'(u)) \chi(v) \big]}{\cal N_\sharp-\zeta}\,,
\end{equation}
where $\chi \in C_c^\infty(\R)$ is a function such that $\chi =1$ on $[-1,1]$.
The identity \eqref{eq:Helffer-Sjostrand_application} can be deduced from the proof of \cite[Proposition C.1]{Knowles_Benaych-Georges} with $n=1$. More precisely, we use the assumption that $f \in C_c^\infty(\R)$ in order to deduce that we can take $\chi$ to be a compactly supported function in the $v$ variable. Furthermore, $\chi$ can be taken to be equal to $1$ on $[-1,1]$ since $\spec (\cal N_\tau) \subset \R$ (c.f.\ \cite[(C.1)]{Knowles_Benaych-Georges}) and since $\cal N$ takes values in $\R$.

Let us define 
\begin{equation*}
\psi(\zeta) \;\deq\;\frac{1}{\pi}\partial_{\bar{\zeta}}\big[(f(u)+\ii vf'(u)) \chi(v) \big]\;=\;\frac{1}{2\pi} \big[\ii v\,f''(u) \, \chi(v)+\ii \big(f(u)+\ii vf'(u)\big)\,\chi'(v)\big]\,,
\end{equation*}
so that by \eqref{eq:Helffer-Sjostrand_application} we have 
\begin{equation}
\label{eq:Function psi}
f(\cal N_\sharp) \;=\;\int_{\C} \dd \zeta\,\frac{\psi(\zeta)}{\cal N_\sharp-\zeta}\,.
\end{equation}
Since $f,\chi \in C_c^\infty(\R)$, it follows that
\begin{equation}
\label{eq:psi compactly supported}
\psi \in C_c^\infty(\C)\,.
\end{equation}
By our choice of $\psi$ we know that
\begin{equation}
\label{eq:psi bound}
|\psi(\zeta)| \;\leq\; C|v| \;=\; C|\im \zeta|\,.
\end{equation}
Substituting this into \eqref{eq:Function psi} we deduce that
\begin{equation}
\label{eq:Function alpha}
\rho_\sharp\big(\Theta_\sharp(\xi) f(\cal N_\sharp)\big) \;=\; \int_{\C} \dd \zeta\,\psi(\zeta)\,\rho_\sharp \bigg(\Theta_\sharp(\xi)\,\frac{1}{\cal N_\sharp-\zeta}\bigg) \;=\;
\int_{\C} \dd \zeta\, \psi(\zeta)\,\alpha_\sharp^\xi(\zeta)\,.
\end{equation}
We note that, by \eqref{eq:Function alpha I5}, \eqref{eq:psi compactly supported} and \eqref{eq:psi bound} we have for almost all $\zeta \in \mathbb{C}$
\begin{equation}
\label{eq:Function F}
|\psi(\zeta)\, \alpha_\tau^\xi(\zeta)| \leq F(\zeta)
\end{equation} 
for some function $F \in L^1(\C)$. Therefore, the interchanging of the integration in $\zeta$ and expectation $\rho_\sharp(\cdot)$ in \eqref{eq:Function alpha} is justified by Fubini's theorem.
Furthermore, recalling \eqref{eq:Function beta} and using \eqref{eq:Function beta claim}, we note that 
\begin{equation}
\label{eq:Function alpha claim}
\lim_{\tau \rightarrow \infty} \sup_{\xi \in \cal C_p} |\alpha_\tau^\xi(\zeta)-\alpha^\xi(\zeta)|=0\,\, \mbox{for all} \,\, \zeta \in \mathbb{C} \setminus [0,\infty)\,.
\end{equation}
The claim (i) now follows from \eqref{eq:Function alpha}--\eqref{eq:Function alpha claim} and the dominated convergence theorem.

We now prove (ii).
Let us define $\gamma_{\sharp,p}^f$ by duality according to
\begin{equation*}
\tr \, \big(\gamma_{\sharp,p}^f \eta\big) \;=\; \rho_\sharp \big(\Theta_\sharp(\eta) f(\cal N_\sharp)\big)\,,
\end{equation*}
for $\eta \in \cal L(\fra H^{(p)})$.
In particular
\begin{equation}
\label{eq:gamma f}
\gamma_{\sharp,p}^f (x_1,\ldots,x_p;y_1,\ldots,y_p) \;=\; \rho_\sharp \big(\phi_\sharp^*(y_1) \cdots \phi_\sharp^*(y_p) \phi_\sharp(x_1) \cdots \phi_\sharp(x_p) f(\cal N_\sharp)\big)\,.
\end{equation}
By duality, part (i) implies that
$\lim_{\tau \to \infty} \|\gamma_{\tau,p}^f - \gamma_p^f\|_{\fra S^2(\fra H^{(p)})}=0$ and 
$\lim_{\tau \to \infty} \tr \gamma_{\tau,p}^f =\tr \gamma_p^f$.
The claim follows from \cite[Lemma 4.10]{FrKnScSo} (which in turn is based on arguments from the proof of \cite[Lemma 2.20]{Simon05}) if we prove that the $\gamma_{\sharp,p}^f$ are positive operators.
Namely, if this is the case, the conclusion of \cite[Lemma 4.10]{FrKnScSo} is that we have $\lim_{\tau \to \infty} \|\gamma_{\tau,p}^f - \gamma_p^f\|_{\fra S^1(\fra H^{(p)})}=0$ and claim (ii) then follows by duality.

We now prove the positivity of $\gamma_{\sharp,p}^f $. Given $\eta \in \fra H^{(p)}$, a direct calculation using \eqref{eq:gamma f} shows that we have
\begin{equation*}
\langle \eta, \gamma_{\sharp,p}^f \eta \rangle_{\fra H^{(p)}} \;=\;\rho_\sharp \big(\Theta_\sharp(\eta \otimes \bar{\eta}) f(\cal N_\sharp)\big)\,.
\end{equation*}
This quantity is nonnegative in the quantum setting since $\Theta_\tau(\eta \otimes \bar{\eta}), f(\cal N_\tau), \ee^{-H_{\tau,0}- \cal W_\tau}$ are positive operators on $\cal F$. In order to see the positivity of $\Theta_\tau(\eta \otimes \bar{\eta})$, we apply \eqref{n_sector}. Moreover, in the classical setting, the quantities
\begin{equation*}
\Theta(\eta \otimes \bar{\eta})
\;=\;\bigg|\int \dd x_1 \cdots \dd x_p \,\bar {\phi}(x_1) \cdots \bar{\phi}(x_p) \, \eta(x_1,\ldots,x_p)\bigg|^2\,, \,\, f(\cal N)\,,\,\, \ee^{-W}
\end{equation*}
are nonnegative. Therefore the $\gamma_{\sharp,p}^f$ are indeed positive operators. Note that this is the only step where we use the nonnegativity of $f$.
\end{proof}

\subsection{Schwinger-Dyson expansion in the quantum problem}

Arguing similarly as in \cite[Section 4.2]{Knowles_Thesis}, we apply a Schwinger-Dyson expansion to $\Psi_{\tau}^t\, \Theta_\tau(\xi)$. Here we recall the time-evolution operator $\Psi_{\tau}^t$ from Definition \ref{Quantum_time_evolution}. We note that a related approach was also applied in \cite{FrKnPi,FrKnSchw}.

Before we proceed with the expansion, we first introduce the operation $\bullet_r$ as well as the free quantum time evolution of operators on $\cal F$, analogously to Definition \ref{Quantum_time_evolution}.
\begin{definition}
\label{Quantum Definition times r}
Let $\xi \in \cal L(\fra H^{(p)})$, $\eta \in \cal L(\fra H^{(q)})$ and $r \leq \min\{p,q\}$ be given.
\begin{itemize}
\item[(i)] We define
\begin{equation*}
\xi \bullet_r \eta \;\deq\, P_{+} (\xi \otimes \mathbf{1}^{(q-r)}) \, (\mathbf{1}^{(p-r)} \otimes \eta) P_{+} \in \cal L(\fra H^{(p+q-r)})\,,
\end{equation*}
where we recall that $P_{+}$ denotes the orthogonal projection from $\fra H^{\otimes r}$ to $\fra H^{(r)}$.
\item[(ii)] With $\bullet_r$ given by (i), we define 
\begin{equation*}
[\xi,\eta]_r \;\deq\; \xi \bullet_r \eta-\eta \bullet_r \xi \in \cal L(\fra H^{(p+q-r)})\,.
\end{equation*}
\end{itemize}
\end{definition}

The following lemma can be found in \cite[Section 3.4.1]{Knowles_Thesis}. We omit the proof.
\begin{lemma}
\label{Quantum Lemma 1}
Let $\xi \in \cal L(\fra H^{(p)})$, $\eta \in \cal L(\fra H^{(q)})$ and $r \leq \min\{p,q\}$ be given.
The following identities hold.
\begin{itemize}
\item[(i)] $\Theta_\tau(\xi) \Theta_\tau(\eta) = \sum_{r=0}^{\min \{p,q\}} {p \choose r} {q \choose r} \frac{r!}{\tau^r} \Theta_\tau (\xi \bullet_r \eta)$.
\item[(ii)] $[\Theta_\tau(\xi),\Theta_\tau(\eta)]=\sum_{r=1}^{\min \{p,q\}}{p \choose r} {q \choose r} \frac{r!}{\tau^r} \Theta_\tau ([\xi,\eta]_r)$.
\end{itemize}
\end{lemma}
\begin{definition}
\label{Quantum_time_evolution2}
Let $\mathbf{A}$ be an operator on $\cal F$. 
We define its free quantum time evolution by
\begin{equation*}
\Psi_{\tau,0}^t \, \mathbf{A} \;\deq\;\ee^{\ii t \tau H_{\tau,0}} \mathbf{A}\, \ee^{-\ii t \tau H_{\tau,0}}\,.
\end{equation*}
\end{definition}
Note that, using first-quantized notation, we have
\begin{equation}
\label{rescaled_Hamiltonian}
\tau H_{\tau,0} \big|_{\fra H^{(n)}} \;=\;\sum_{i=1}^{n}h_i\,.
\end{equation}
Here $h_i$ denotes the operator $h$ acting in the $x_i$ variable.
By \eqref{rescaled_Hamiltonian}, we note  the operator $\Psi_{\tau,0}^t$ does not depend on $\tau$. We keep the subscript $\tau$ in order to emphasize that this is a quantum time evolution.
Moreover, it is useful to apply a time evolution to $p$-particle operators.
\begin{definition}
\label{xi_time_evolution}
Let $\xi \in \cal L(\fra H^{(p)})$. For $t \in \mathbb{R}$ we define
\begin{equation*}
\xi_{t} \;\deq\; \ee^{\ii t \sum_{j=1}^{p} h_j} \, \xi \, \ee^{-\ii t \sum_{j=1}^{p} h_j} \,.
\end{equation*}
\end{definition}
In particular, from \eqref{Theta_tau_xi}, Definition \ref{Quantum_time_evolution2}, \eqref{rescaled_Hamiltonian} and Definition \ref{xi_time_evolution}, it follows that 
for $\xi \in \cal L(\fra H^{(p)})$ we have
\begin{equation}
\label{Psi_0_xi}
\Psi_{\tau,0}^t\, \Theta_\tau(\xi) \;=\; \Theta_\tau(\xi_t)\,.
\end{equation}
The following result holds.
\begin{lemma}
\label{Schwinger_Dyson_expansion_Q}
Let $\xi \in \cal L(\fra H^{(p)})$.
Given $\cal K >0$, $\epsilon>0$ and $t \in \R$, there exists $L=L(\cal K,\epsilon,t,\|\xi\|,p) \in \N$, a finite sequence $(e^l)_{l=0}^{L}$ with $e^l = e^l(\xi,t) \in \cal L(\fra H^{(l)})$ and $\tau_0=\tau_0(\cal K,\epsilon,t,\|\xi\|)>0$ such that
\begin{equation}
\label{Schwinger_Dyson_expansion_Q_estimate}
\Bigg\|\bigg(\Psi_{\tau}^t\Theta_\tau(\xi)
-\sum_{l=0}^{L}\Theta_\tau(e^l)\bigg)\bigg|_{\fra H^{(\leq \cal K \tau)}}\Bigg\| \;\leq\; \epsilon\,,
\end{equation}
for all $\tau \geq \tau_0$. Note that $\fra H^{(\leq \cal K \tau)}$ is defined as in \eqref{H(R)} above.
\end{lemma}

\begin{proof}
Let us first observe that
\begin{equation}
\label{Schwinger_Dyson_Expansion_Q1}
\Psi_{\tau}^t\, \Theta_\tau(\xi) \;=\;\Theta_\tau(\xi_t)+(\ii p) \, \int_0^t \dd s \, \Psi_{\tau}^s \,\Theta_\tau\big(\big[W,\xi_{t-s}\big]_1\big)+ \frac{\ii {p \choose 2}}{\tau} \, \int_0^t \dd s \, \Psi_{\tau}^s \,\Theta_\tau\big(\big[W,\xi_{t-s}\big]_2\big)\,.
\end{equation}
Indeed, we write
\begin{equation}
\label{Schwinger_Dyson_Expansion_int_step1A}
\Psi_{\tau}^t\, \Theta_\tau(\xi) \;=\;\Psi_{\tau}^s \,\Psi_{\tau,0}^{-s}\,\Psi_{\tau,0}^t \, \Theta_\tau(\xi) \big|_{s=t} \;=\;
\Psi_{\tau,0}^t \, \Theta_\tau(\xi) + \int_0^t \,\dd s \,\frac{\dd}{\dd s} \Big(\Psi_{\tau}^s \,\Psi_{\tau,0}^{-s}\,\Psi_{\tau,0}^t \, \Theta_\tau(\xi)\Big)\,,
\end{equation}
which by \eqref{Psi_0_xi} and Definitions \ref{Quantum_time_evolution} and  \ref{Quantum_time_evolution2} equals
\begin{equation}
\label{Schwinger_Dyson_Expansion_int_step1}
\Theta_\tau(\xi_t)+\int_0^t \,\dd s \,\frac{\dd}{\dd s} \Big(\ee^{\ii s \tau H_\tau} \,\ee^{-\ii s\tau H_{\tau,0}} \,\Theta_\tau(\xi_t)\,\ee^{\ii s \tau H_{\tau,0}} \,\ee^{-\ii s \tau H_\tau} \Big)\,.
\end{equation}
By differentiating in $s$ and using \eqref{H_tau}, \eqref{H_tau_0} and \eqref{Psi_0_xi}, it follows that the integrand in the second term of \eqref{Schwinger_Dyson_Expansion_int_step1} equals
\begin{equation*}
\frac{\ii \tau}{2} \Psi_{\tau}^s \, \Psi_{\tau,0}^{-s} \big[\Theta_\tau(W_s),\Theta_\tau(\xi_t)\big] \;=\; \frac{\ii \tau}{2} \Psi_{\tau}^s \, \big[\Theta_\tau(W),\Theta_\tau(\xi_{t-s})\big]\,,
\end{equation*}
which by Lemma \ref{Quantum Lemma 1} (ii) equals
\begin{equation}
\label{Schwinger_Dyson_Expansion_int_step2}
(\ii p)\, \Theta_\tau\big(\big[W,\xi_{t-s}\big]_1\big)+\frac{\ii {p \choose 2}}{\tau} \,\Theta_\tau\big(\big[W,\xi_{t-s}\big]_2\big)\,.
\end{equation}
Substituting \eqref{Schwinger_Dyson_Expansion_int_step2} into \eqref{Schwinger_Dyson_Expansion_int_step1}, we deduce \eqref{Schwinger_Dyson_Expansion_Q1}.

Iteratively applying \eqref{Schwinger_Dyson_Expansion_Q1} we deduce that, for all $M \in \N$ we have
\begin{equation*}
\Psi_{\tau}^t\, \Theta_\tau(\xi) \;=\;A_{\tau,M}^t(\xi)+E_{\tau,M}^t(\xi)+B_{\tau,M}^t(\xi)\,,
\end{equation*}
where
\begin{multline}
\label{A_E_B}
A_{\tau,M}^t(\xi) \;\deq\;\Theta_\tau(\xi_t)+ \sum_{j=1}^{M-1}\ii^j \, p \,(p+1)\, \cdots \,(p+j-1) \bigg\{ \int_0^t \dd s_1\,\int_0^{s_1} \dd s_2 \,\cdots \int_0^{s_{j-1}} \dd s_j \, 
\\
\Theta_\tau\Big(\big[W_{s_j},\big[W_{s_{j-1}}, \ldots, \big[W_{s_1},\xi_t\big]_1 \ldots \big]_1\big]_1\Big) \bigg\}\,,
\\
E_{\tau,M}^t(\xi) \;\deq\; \ii^M \,p \,(p+1)\, \cdots \, (p+M-1) \, \bigg\{\int_0^t \dd s_1 \, \int_0^{s_1} \dd s_2 \, \cdots \int_0^{s_{M-1}} \dd s_M\, 
\\
\Psi_{\tau}^{s_M} \, \Theta_\tau \Big(\big[W,\big[W_{s_{M-1}-s_M}, \ldots, \big[W_{s_1-s_M},\xi_{t-s_M}\big]_1 \ldots \big]_1\big]_1\Big)\bigg\}\,,
\\
B_{\tau,M}^t(\xi) \;\deq\; \frac{1}{\tau} \, \sum_{j=1}^M \ii^j \, p \, (p+1) \, \cdots \, (p+j-2) \, {p+j-1 \choose 2} \, \bigg\{\int_0^t \dd s_1 \, \int_0^{s_1} \dd s_2 \, \cdots \int_0^{s_{j-1}} \dd s_j\, 
\\
\Psi_{\tau}^{s_j} \, \Theta_\tau \Big(\big[W,\big[W_{s_{j-1}-s_j}, \ldots, \big[W_{s_1-s_j},\xi_{t-s_j}\big]_1 \ldots \big]_1\big]_2\Big)\bigg\}\,.
\end{multline} 
Moreover, we define $A_{\tau,\infty}^t(\xi)$ and $B_{\tau,\infty}^t(\xi)$ by formally setting $M=\infty$ in \eqref{A_E_B}. We now show that, on $\fra H^{(\leq \cal K \tau)}$ we have 
\begin{equation}
\label{n_assumption_convergence}
A_{\tau,M}^t(\xi)  \rightarrow A_{\tau,\infty}^t(\xi)\,,\qquad E_{\tau,M}^t(\xi) \rightarrow 0\,,\qquad B_{\tau,M}^t(\xi)  \rightarrow B_{\tau,\infty}^t(\xi)
\end{equation}
as $M \rightarrow \infty$ in norm whenever $|t|<T_0(\cal K)$, where $T_0(\cal K)$ is chosen sufficiently small depending on $\cal K$, but \emph{independent of} $p$. In particular, it follows that on $\fra H^{(\leq \cal K \tau)}$, the formally-defined quantities $A_{\tau,\infty}^t(\xi)$ and $B_{\tau,\infty}^t(\xi)$ are well defined and that $E_{\tau,\infty}^t(\xi)$ vanishes.

In order to prove \eqref{n_assumption_convergence} we note that, if $n \leq \cal K \tau$, the $j$-th term of the formal sum $A_{\tau,\infty}^t(\xi)$ acting on $\fra H^{(n)}$ is estimated in norm by
\begin{equation}
\label{n_assumption_convergence1}
\frac{|t|^j}{j!} \,(p+j)^j \, 2^j \, \Big(\frac{n}{\tau}\Big)^{p+j} \, \|w\|_{L^\infty}^j \, \|\xi\|\,.
\end{equation}
Here we used Lemma \ref{Quantum Lemma 2} as well as $\|\xi_t\|=\|\xi\|$, $\|W_s\|=\|W\|=\|w\|_{L^\infty}$. The latter two equalities follow immediately from Definition \ref{xi_time_evolution}. The expression in \eqref{n_assumption_convergence1} is
\begin{equation}
\label{n_assumption_convergence2}
\;\leq\; \ee^p \, \cal K^p \, \Big( 2 \ee \, \cal K \, \|w\|_{L^\infty} \,|t|\Big)^j \, \|\xi\| \,.
\end{equation} 
Using \eqref{n_assumption_convergence2}, we can deduce the first convergence result in \eqref{n_assumption_convergence} for $|t|<T_0(\cal K)$. 
By noting that $\Psi_{\tau}^s$ preserves the operator norm, we deduce the second and third convergence results in \eqref{n_assumption_convergence} by an analogous argument.
 We omit the details.
In particular, on $\fra H^{(\leq \cal K \tau)}$ we can write for $|t| < T_0(\cal K)$
\begin{equation}
\label{A+B}
\Psi_{\tau}^t\, \Theta_\tau(\xi) \;=\;A_{\tau,\infty}^t(\xi)+B_{\tau,\infty}^t(\xi)\,,
\end{equation}
where the infinite sum converges in norm. Recalling \eqref{A_E_B}, it also follows from this proof that
\begin{equation}
\label{B_bound}
\big\|B_{\tau,\infty}^t(\xi)\big|_{\fra H^{(\leq \cal K \tau)}}\big\| \;\leq\;\frac{C \ee^{p} \,\cal K^p \,\|\xi\|}{\tau}\,.
\end{equation}
By \eqref{A_E_B}--\eqref{n_assumption_convergence}, \eqref{A+B}--\eqref{B_bound}, we deduce that 
\eqref{Schwinger_Dyson_expansion_Q_estimate} holds for $|t|<T_0(\cal K)$.
Note that the $e^l$ are obtained from the partial sums of $A_{\tau,\infty}^t(\xi)$ (as in \eqref{A_E_B}).
By construction we have that $e^l \in \cal L(\fra H^{(l)})$.

We obtain \eqref{Schwinger_Dyson_expansion_Q_estimate} for general $t$ by iterating this procedure in increments of size $T_0(\cal K)$. This is possible to do by using norm conservation, i.e.\ we use that for all operators $\bf A$ on $\cal F$ we have
\begin{equation}
\label{norm_conservation_Q}
\big\|\Psi_{\tau}^t\,\bf A \big|_{\fra H^{(\leq \cal K \tau)}}\big\| \;=\;\big\|\bf A \big|_{\fra H^{(\leq \cal K \tau)}}\big\|\,.
\end{equation}
Furthermore, we use the observation that the radius of convergence $T_0(\cal K)$ does not depend on $p$. The latter fact is required since after each iteration of the procedure we generate $q$-particle operators, where $q$ grows with $t$. A detailed description of an analogous iteration procedure applied in a slightly different context can be found in \cite[Lemma 3.6]{Knowles_Thesis}.
\end{proof}

\subsection{Schwinger-Dyson expansion in the classical problem}
The following lemma can be found in \cite[Section 3.4.2]{Knowles_Thesis}. We omit the proof.
\begin{lemma}
\label{Classical Lemma 1}
Let $\xi \in \cal L(\fra H^{(p)})$, $\eta \in \cal L(\fra H^{(q)})$ be given. We then have 
\begin{equation*}
\big\{\Theta(\xi),\Theta(\eta)\big\}\;=\;\ii \, pq \,\Theta\big([\xi,\eta]_1\big)\,.
\end{equation*}
\end{lemma}
\begin{definition}
\label{Classical_time_evolution2}
Let $\xi \in \cal L(\fra H^{(p)})$. 
We define $\Psi^t_0 \, \Theta(\xi)$ to be the random variable
\begin{equation*}
\int \dd x_1 \cdots \dd x_p \, \dd y_1 \cdots \dd y_p \, \xi(x_1, \dots, x_p; y_1, \dots, y_p) \, \overline{S_{t,0} \phi}(x_1) \cdots \overline{S_{t,0}\phi}(x_p) \,S_{t,0}\phi(y_1) \cdots S_{t,0}\phi(y_p)\,,
\end{equation*} 
where $S_{t,0} \deq \ee^{-\ii t h}$ denotes the free Schr\"{o}dinger evolution on $\fra H$ 
corresponding to the Hamiltonian \eqref{Hamiltonian h1}.
\end{definition}
In particular, from Definitions \ref{Classical_time_evolution} and \ref{Classical_time_evolution2} we have 
\begin{equation}
\label{Poisson bracket}
\partial_t \Psi^t \, \Theta(\xi)\;=\; \Psi^t \, \{H, \Theta(\xi)\}\,,\qquad\,\partial_t \Psi_0^t \,\Theta(\xi) \;=\;  \Psi_0^t \,\{H_0,\Theta(\xi)\}\,,
\end{equation} 
where we recall \eqref{classical_hamiltonian}--\eqref{classical_free_hamiltonian}.

From \eqref{theta_xi}, Definition \ref{xi_time_evolution} and Definition \ref{Classical_time_evolution2}, it follows that 
for $\xi \in \cal L(\fra H^{(p)})$ we have
\begin{equation}
\label{Psi_0_xi_C}
\Psi_{0}^t \, \Theta(\xi) \;=\; \Theta(\xi_t)\,.
\end{equation}
We now prove the classical analogue of Lemma \ref{Schwinger_Dyson_expansion_Q}.

\begin{lemma}
\label{Schwinger_Dyson_expansion_C}
Let $\xi \in \cal L(\fra H^{(p)})$.
Given $\cal K >0$, $\epsilon>0$ and $t \in \R$, for $L=L(\cal K,\epsilon,t,\|\xi\|,p) \in \N$ and $\tau_0=\tau_0(\cal K,\epsilon,t,\|\xi\|)>0$ chosen possibly larger than in Lemma \ref{Schwinger_Dyson_expansion_Q} and for the same choice of $e^l = e^l(\xi,t) \in \cal L(\fra H^{(l)})$ as in Lemma \ref{Schwinger_Dyson_expansion_Q} we have
\begin{equation*}
\Bigg|\Bigg(\Psi^t \Theta(\xi)
-\sum_{l=0}^{L}\Theta(e^l)\Bigg) \; \mathbf{1}_{\{\cal N \leq \cal K\}}\Bigg| \;\leq\; \epsilon\,,
\end{equation*}
for all $\tau \geq \tau_0$.
\end{lemma}

\begin{proof}
We first note that we have the classical analogue of \eqref{Schwinger_Dyson_Expansion_Q1}
\begin{equation}
\label{Schwinger_Dyson_Expansion_C1}
\Psi^t \, \Theta(\xi) \;=\;\Theta(\xi_t)+(\ii p) \, \int_0^t \dd s \, \Psi^s \,\Theta\big(\big[W,\xi_{t-s}\big]_1\big)\,.
\end{equation}
Namely, arguing as in \eqref{Schwinger_Dyson_Expansion_int_step1A} and using \eqref{Psi_0_xi_C}, it follows that
\begin{equation}
\label{Schwinger_Dyson_Expansion_C2}
\Psi^t \, \Theta(\xi) \;=\;
\Theta(\xi_t) + \int_0^t \,\dd s \,\frac{\dd}{\dd s} \Big(\Psi^s \,\Psi_0^{-s}\,\Psi_0^t \, \Theta(\xi)\Big)\,.
\end{equation}
Differentiating and using \eqref{Poisson bracket}, it follows that the integrand in \eqref{Schwinger_Dyson_Expansion_C2} equals
\begin{equation*}
\Psi^s \big\{H,\Psi_0^{-s+t} \,\Theta(\xi)\big\}-\Psi^s \Psi_0^{-s} \big\{H_0,\Psi_0^t \,\Theta(\xi)\big\} 
\;=\;\Psi^s \big\{H,\Theta(\xi_{-s+t})\big\}-\Psi^s \Psi_0^{-s} \big\{H_0,\Theta(\xi_t)\big\}\,.
\end{equation*}
In the last equality we also used \eqref{Psi_0_xi_C}.
By Lemma \ref{Classical Lemma 1} and \eqref{classical_hamiltonian}, we can rewrite this as
\begin{equation}
\label{Schwinger_Dyson_Expansion_C3}
\Psi^s \Big(\ii p \, \Theta\big(\big[h+W,\xi_{-s+t}\big]_{1}\big)\Big)-\Psi^s \Big(\Psi_0^{-s} \,\ii p \,\Theta\big(\big[h,\xi_t\big]_1 \big)\Big)\,.
\end{equation}
We note that $\Psi_0^{-s} \Theta \big(\big[h,\xi_t\big]_1\big)=\Theta \big(\big[h,\xi_{-t+s}\big]_1\big)$ and hence the expression in \eqref{Schwinger_Dyson_Expansion_C3} equals
\begin{equation*}
\ii p \, \Psi^s \,\Theta\big(\big[W,\xi_{t-s}\big]_1\big)\,.
\end{equation*}
Substituting this into \eqref{Schwinger_Dyson_Expansion_C2} we obtain \eqref{Schwinger_Dyson_Expansion_C1}.

We now iterate \eqref{Schwinger_Dyson_Expansion_C1} analogously as in the proof of Lemma \ref{Schwinger_Dyson_expansion_Q}.
The convergence for $|t| < T_0(\cal K)$ is shown by arguing as in the proof of \eqref{n_assumption_convergence}. The only difference is that instead of applying Lemma \ref{Quantum Lemma 2}, we now apply Lemma \ref{Classical Lemma 2}. 
(In fact, the quantity $T_0(\cal K)$ can be chosen to be the same as the corresponding quantity in Lemma \ref{Schwinger_Dyson_expansion_Q}, which was obtained from \eqref{n_assumption_convergence2}).
Furthermore, in the extension to all times, instead of applying \eqref{norm_conservation_Q}, we use that $S_t$ preserves the norm on $\fra H$. 
Finally, we note that the $e^l$ that we obtain from iterating \eqref{Schwinger_Dyson_Expansion_C1} are the same as those obtained by iterating \eqref{Schwinger_Dyson_Expansion_Q1} in the proof of Lemma \ref{Schwinger_Dyson_expansion_Q}.
\end{proof}

\subsection{Proof of Proposition \ref{Small_N_convergence}}
\label{Proof of Small_N_convergence}
We now combine the results of Proposition \ref{Convergence 1}, Lemma \ref{Schwinger_Dyson_expansion_Q} and Lemma \ref{Schwinger_Dyson_expansion_C} in order to prove Proposition \ref{Small_N_convergence}.

\begin{proof}[Proof of Proposition \ref{Small_N_convergence}]
By assumption, there exists $\cal K>0$ such that $F=0$ on $(\cal K,\infty)$.
Let us note that, for all $\xi \in \cal L(\fra H^{(p)})$ and for all $t \in \R$, the following  inequalities hold in the quantum setting.
\begin{equation}
\label{small_N_convQ1}
\Big\|\Psi_{\tau}^t\Theta_\tau(\xi)\big|_{\fra H^{(\leq \cal K \tau)}}\Big\| \;\leq\; \cal K^p\,\|\xi\|\,.
\end{equation}
\begin{equation}
\label{small_N_convQ2}
\frac{\|F(\cal N_\tau)\,\ee^{-H_\tau}\|_{\fra S^1(\cal F)}}{\tr (\ee^{-H_\tau})} \;\leq\; C\,,
\end{equation}
for some constant $C>0$ independent of $\tau$. 
The inequality \eqref{small_N_convQ1} follows from Definition \ref{Quantum_time_evolution}, the observation that  
$\Psi_{\tau}^t$  preserves operator norm and from Lemma \ref{Quantum Lemma 2}. The inequality \eqref{small_N_convQ2} follows from $F(\cal N_\tau) \ee^{-H_\tau} \geq 0$ and $\frac{\tr (F(\cal N_\tau) \ee^{-H\tau})}{\tr (\ee^{-H_\tau})} \leq C$. Furthermore, in the classical setting, the following inequalities hold.
\begin{equation}
\label{small_N_convC1}
\big|\Psi^t \Theta(\xi)\big| \;\leq\; \cal K^p \,\|\xi\|\,,\quad\mbox{whenever}\quad \|\phi\|_{\fra H}^2 \;\leq\; \cal K\,.
\end{equation}
\begin{equation}
\label{small_N_convC2}
\rho(F(\cal N)) \;\leq\; C\,.
\end{equation}
The inequality \eqref{small_N_convC1} follows from Definition \ref{Classical_time_evolution}, the Cauchy-Schwarz inequality and since $S_t$ preserves the norm on $\fra H$.
The inequality \eqref{small_N_convC2} is immediate since $F \in C_c^{\infty}(\R)$.

We now apply H\"{o}lder's inequality, \eqref{small_N_convQ1}-\eqref{small_N_convC2}
and Lemmas \ref{Schwinger_Dyson_expansion_Q} and \ref{Schwinger_Dyson_expansion_C} with $\cal K$ chosen as above to deduce that the claim follows if we prove that for all $q_1,\ldots,q_m \in \N$, $\eta^1\in \cal L(\fra H^{(q_1)}),\ldots,\eta^m\in \cal L(\fra H^{(q_m)})$ we have
\begin{equation}
\label{Small_N_convergence_claim}
\lim_{\tau \to \infty} \rho_\tau \big(\Theta_\tau(\eta^1)\,\cdots\, \Theta_\tau(\eta^m)\,F(\cal N_\tau)\big) \;=\;
\rho \big(\Theta(\eta^1)\,\cdots\, \Theta(\eta^m)\,F(\cal N)\big)\,.
\end{equation}
Note that in the iterative application of \eqref{small_N_convQ1} we use that the operator $\Psi_{\tau}^t\Theta_\tau(\xi)$ leaves the sectors $\fra H^{(n)}$ of the Fock space invariant, thus allowing us to apply the estimate on $\fra H^{(\leq \cal K\tau)}$. 
We can rewrite the right-hand side of \eqref{Small_N_convergence_claim} as 
\begin{equation}
\label{Small_N_convergence_claim1}
\rho \big(\Theta(\eta)F(\cal N)\big)\,,
\end{equation}
where $q \deq q_1+ \cdots + q_m$ and
\begin{equation}
\label{eta_definition}
\eta \;\deq\; \eta^1\bullet_0 \cdots \bullet_0  \eta^m\in \cal L(\fra H^{(q)})\,.
\end{equation}
Here we recall Definition \ref{Quantum Definition times r} (i). 

By iteratively applying Lemma \ref{Quantum Lemma 1} (i) and using H\"{o}lder's inequality together with \eqref{small_N_convQ1}--\eqref{small_N_convQ2}, it follows that the left-hand side of \eqref{Small_N_convergence_claim} equals
\begin{equation}
\label{Small_N_convergence_claim2}
\rho_\tau \big(\Theta_\tau(\eta)F(\cal N_\tau)\big)+\cal  O\bigg(\frac{1}{\tau}\bigg)\,,
\end{equation}
where $\eta$ is given by \eqref{eta_definition} above.
The convergence \eqref{Small_N_convergence_claim} now follows from \eqref{Small_N_convergence_claim1}-\eqref{Small_N_convergence_claim2}, the assumptions on $F$ and Proposition \ref{Convergence 1} (ii).
\end{proof}

\section{The large particle number regime: proof of Proposition \ref{Large_N_convergence}.}
\label{Large particle number}

In this section we consider the regime where $\cal N_\tau$, $\cal N$ are assumed to be large. The main result that we prove is Proposition \ref{Large_N_convergence}.

\begin{proof}[Proof of Proposition \ref{Large_N_convergence}]
We first prove (i). 
Let us note that $\cal N_\tau$ commutes with $\ee^{-H_\tau}$ and with $\Psi_\tau^{t_j}(\xi^j)$ for all $j=1,\ldots,m$. Therefore, the expression that we want to estimate in (i) can be rewritten as
\begin{equation*}
\Big|\rho_\tau  \Big(\big(1+\cal N_\tau\big)^{-p_1}\,\Psi_\tau^{t_1} \,\Theta_\tau(\xi^1)\,\cdots \, 
\big(1+\cal N_\tau\big)^{-p_m}\,\Psi_\tau^{t_m} \,\Theta_\tau(\xi^m)\,\big(1+\cal N_\tau\big)^{p}\,G(\cal N_\tau)
\Big)\Big|\,,
\end{equation*}
where we define $p \deq p_1+\cdots+p_m$.
Using H\"{o}lder's inequality and $(1+\cal N_\tau)^p\,G(\cal N_\tau) \geq 0$, this is
\begin{equation}
\label{Large_N_convergence1}
\;\leq\; \Bigg(\prod_{j=1}^{m}\big\|\big(1+\cal N_\tau\big)^{-p_j}\,\Psi_\tau^{t_j} \,\Theta_\tau (\xi^j)\big\|\Bigg) \, \rho_\tau \Big(\big(1+\cal N_\tau\big)^{p}\,G(\cal N_\tau)
\Big)\,.
\end{equation}
The $j$-th factor of the first expression in \eqref{Large_N_convergence1} equals
\begin{equation}
\label{Large_N_convergence2}
\big\|\Psi_\tau^{t_j} \,\big(1+\cal N_\tau\big)^{-p_j}\,\Theta_\tau(\xi^j)\big\|
\;=\;\big\|\big(1+\cal N_\tau\big)^{-p_j}\,\Theta_\tau(\xi^j)\big\| \;\leq\; \|\xi^j\|\,.
\end{equation}
Here we used that $\cal N_\tau$ commutes with $\ee^{\ii t \tau H_\tau}$, that $\Psi_\tau^{t_j}$ preserves operator norm and Lemma \ref{Quantum Lemma 2}.

By construction of $G$ we note that the second expression in \eqref{Large_N_convergence1} is
\begin{equation*}
\;\leq\; \rho_\tau\Big((1+\cal N_\tau)^p\, \mathbf{1}(\cal N_\tau \;\geq\; \cal K)\Big)\,,
\end{equation*}
which by Markov's inequality is 
\begin{equation}
\label{Large_N_convergence3}
\;\leq\;\frac{\rho_\tau \big((1+\cal N_\tau)^{p+1}\big)}{\cal K} \;\leq\; \frac{C(p)}{\cal K}\,.
\end{equation}
The above application of Markov's inequality is justified since $\cal N_\tau$ commutes with $\ee^{-H_\tau}$.
Claim (i) now follows from \eqref{Large_N_convergence1}--\eqref{Large_N_convergence3}.

We now prove (ii) by similar arguments. Namely, we rewrite the expression that we want to estimate in (ii) as
\begin{equation*}
\Big|\rho \Big(\big(1+\cal N\big)^{-p_1}\,\Psi^{t_1} \,\Theta(\xi^1)\,\cdots \, 
\big(1+\cal N\big)^{-p_m}\,\Psi^{t_m} \,\Theta(\xi^m)\,\big(1+\cal N\big)^{p}\,G(\cal N)
\Big)\Big|\,, 
\end{equation*}
which is
\begin{equation}
\label{Large_N_convergence4}
\;\leq\;\Bigg(\prod_{j=1}^{m}\Big|\big(1+\cal N\big)^{-p_j}\,\Psi^{t_j} \,\Theta(\xi^j)\Big|\Bigg) \, \rho \Big(\big(1+\cal N\big)^{p}\,G(\cal N)
\Big)\,.
\end{equation}
Using the observation that $S_{t_j}$ preserves the norm on $\fra H$ as well as Lemma \ref{Classical Lemma 2}, it follows that the $j$-th factor of the first term in \eqref{Large_N_convergence4} is bounded by $\|\xi^j\|$. We again use the properties of $G$ and Markov's inequality to deduce that the second term in \eqref{Large_N_convergence4} is 
\begin{equation*}
\;\leq\; \rho\Big((1+\cal N)^p\, \mathbf{1}(\cal N \;\geq\; \cal K)\Big) \;\leq\;\frac{\rho \big((1+\cal N)^{p+1}\big)}{\cal K} \;\leq\; \frac{C(p)}{\cal K}\,.
\end{equation*}
Claim (ii) now follows as in the quantum setting.
\end{proof}

\begin{remark}
\label{uniformity_remark} 
Following the proofs of Proposition \ref{Small_N_convergence} and \ref{Large_N_convergence}, it is immediate that the convergence in Theorem \ref{Main Result} is uniform on the set of parameters $w \in L^\infty(\Lambda)$, $t_1\in \mathbb{R},\ldots,t_m \in \R$, $p_1,\ldots,p_m \in \N$, $m \in \N$, satisfying
\begin{equation*}
\text{max}(\|w\|_{L^\infty}, |t_1|, \ldots,|t_m|, p_1, \ldots, p_m, \|\xi^1\|,\ldots,\|\xi^m\|, m) \;\leq \; M\,,
\end{equation*}
for any fixed $M > 0$.
\end{remark}

\section{The local problem.}
\label{The local problem}
In this section we fix 
\begin{equation*}
\Lambda=\mathbb{T}^1 \quad \mbox{and} \quad v=0\,.
\end{equation*}
Throughout this section and Appendix \ref{Xsb space appendix}, given $s\in \mathbb{R}$, we write $H^s(\Lambda)$ for the $L^2$-based inhomogeneous Sobolev space of order $s$ on $\Lambda$.

We extend the previous analysis to the setting of the local problem \eqref{localNLS1}. In particular, we give the proof of Theorem \ref{Main Result_local_nonlinearity}. 
Before proceeding with the proof of Theorem \ref{Main Result_local_nonlinearity} we prove the following stability result.
\begin{proposition}
\label{L2 convergence}
Let $s \geq \frac{3}{8}$ be given. Let $\phi_0 \in H^s(\Lambda)$. We consider the  Cauchy problem on $\Lambda$ given by
\begin{equation}
\label{NLS}
\begin{cases}
\ii \partial_t u + (\Delta-\kappa)u\;=\;|u|^2u \\
u|_{t=0}\;=\;\phi_0\,.
\end{cases}
\end{equation}
In addition, given $\epsilon>0$, and recalling the definition of $w^\epsilon$ from \eqref{w_epsilon} we consider
\begin{equation}
\label{NLS_epsilon}
\begin{cases}
\ii \partial_t u^{\epsilon}+ (\Delta-\kappa)u^{\epsilon}\;=\;(w^\epsilon*|u^{\epsilon}|^2)u^{\epsilon} \\
u^{\epsilon}|_{t=0}\;=\;\phi_0\,.
\end{cases}
\end{equation}
Let $u,u^{\epsilon}$ be solutions of \eqref{NLS} and \eqref{NLS_epsilon} respectively. Then, for all $T >0 $ we have 
\begin{equation}
\label{L2 convergence claim}
\lim_{\epsilon \to 0} \|u^{\epsilon}-u\|_{L^{\infty}_{[-T,T]}\fra H} \;=\; 0\,.
\end{equation}
\end{proposition}
In order to prove Proposition \ref{L2 convergence} we need to recall several tools from harmonic analysis. In particular, it is helpful to use periodic Strichartz estimates formulated in $X^{\sigma,b}$ spaces. In the context of dispersive PDEs, these spaces were first used in \cite{Bourgain_1993}.
\begin{definition}
\label{X^{sigma,b}}
Given $f: \Lambda \times \R \rightarrow \C$ and $\sigma,b \in \R$ we define 
\begin{equation*}
\|f\|_{X^{\sigma,b}} \;=\; \Big\|\big(1+|2\pi k|\big)^{\sigma}\,\big(1+|\eta+2\pi k^2|\big)^b\, \tilde{f}\Big\|_{L^2_{\eta} l^2_k}\,,
\end{equation*}
where 
\begin{equation*}
\tilde{f}(k,\eta) \;\deq\; \int_{-\infty}^{\infty} \dd t\,  \int_\Lambda \dd x \, f(x,t) \,\ee^{-2\pi \ii k x - 2\pi \ii \eta t }\,
\end{equation*}
denotes the \emph{spacetime Fourier transform}.
\end{definition}
Note that, in particular, we have
\begin{equation*}
\|f\|_{X^{\sigma,b}}\;\sim\;\|\ee^{-\ii t \Delta}f\|_{H^b_t H^{\sigma}_x}\,.
\end{equation*}
Here we use the convention\footnote{We do not introduce additional factors of $2\pi$ in the definition of $\|h\|_{H^b_t}$ for simplicity of notation in the sequel.} that for $h: \mathbb{R} \rightarrow \mathbb{C}$
\begin{equation*}
\|h\|_{H^b_t} \;\deq\; \bigg(\int \dd \eta\, (1+|\eta|)^{2b} |\hat{h}(\eta)|^2\bigg)^{1/2}\,.
\end{equation*}
We now collect several known facts about $X^{\sigma,b}$ spaces. For a more detailed discussion we refer the reader to \cite{Tao}[Section 2.6] and the references therein. For the remainder of this section we fix
\begin{equation}
\label{definition of b}
b \;\deq\;\frac{1}{2}+\nu\,,
\end{equation}
for $\nu>0$ small.

\begin{lemma}
\label{Xsb space lemma}
Let $\sigma \in \R$ and $b$ as in \eqref{definition of b} be given. The following properties hold.
\begin{itemize}
\item[(i)]
$\|f\|_{L^\infty_t H^\sigma_x} \;\leq\; C(b)\,\|f\|_{X^{\sigma,b}}$.
\item[(ii)] Suppose that $\psi \in C_c^{\infty}(\R)$. Then, for all $\delta \in (0,1)$ and $\Phi \in H^\sigma$ we have
\begin{equation*}
\big\|\psi(t/\delta)\,\ee^{\ii t \Delta}\Phi \big\|_{X^{\sigma,b}} \;\leq\;C(b,\psi)\,\delta^{\frac{1-2b}{2}}\|\Phi\|_{H^\sigma}\,.
\end{equation*}
\item[(iii)] Let $\psi, \delta$ be as in (ii). Then, for all $f \in X^{\sigma,b}$ we have
\begin{equation*}
\big\|\psi(t/\delta)f\big\|_{X^{\sigma,b}} \;\leq\;C(b,\psi)\,\delta^{\frac{1-2b}{2}}\,\|f\|_{X^{\sigma,b}}\,.
\end{equation*}
\item[(iv)]
With the same assumptions as in (iii) we have
\begin{equation*}
\bigg\|\psi(t/\delta)\,\int_0^t \dd t'\,\ee^{\ii (t-t') \Delta}\,f(t')\bigg\|_{X^{\sigma,b}} \;\leq\;C(b,\psi)\,\delta^{\frac{1-2b}{2}}\,\|f\|_{X^{\sigma,b-1}}\,.
\end{equation*}
\item[(v)] $\|f\|_{L^4_{t,x}} \;\leq\;C\,\|f\|_{X^{0,3/8}}$.
\item[(vi)] $\|f\|_{X^{0,-3/8}} \;\leq\; C\,\|f\|_{L^{4/3}_{t,x}}$.
\item[(vii)] Let $\psi, \delta$ be as in (ii). Then, for all $f \in X^{0,b}$ we have
\begin{equation*}
\big\|\psi(t/\delta)f\big\|_{L^4_{t,x}} \;\leq\;C(b,\psi)\,\delta^{\theta_0}\,\|f\|_{X^{0,b}}\,,
\end{equation*}
for some $\theta_0 \equiv \theta_0(b)>0$.
\end{itemize}
\end{lemma}
For completeness we present a self-contained proof of Lemma \ref{Xsb space lemma} in Appendix \ref{Xsb space appendix}. 

We also recall the following characterization of homogeneous Sobolev spaces on the torus.
\begin{lemma}
\label{Sobolev space torus}
For $\sigma \in (0,1)$ we have
\begin{equation}
\label{Sobolev space torus norm}
 \Bigg\|\frac{f(x)-f(y)}{[x-y]^{\sigma+\frac{1}{2}}}\Bigg\|_{L^2_{x,y}} \;\sim\; \|f\|_{\dot{H}^\sigma}\,.
\end{equation}
Here $\|f\|_{\dot{H}^\sigma} = \||\nabla|^{\sigma}f\|_{L^2}$ denotes the homogeneous $L^2$-based Sobolev (semi)norm of order $\sigma$.
\end{lemma}
The quantity on the left-hand side of \eqref{Sobolev space torus norm} is the periodic analogue of the \emph{Sobolev-Slobodeckij norm}.
This is a general fact. A self-contained proof using the Plancherel theorem can be found in \cite[Proposition 1.3]{Benyi_Oh}. We now have all the tools to prove Proposition \ref{L2 convergence}.

\begin{proof}[Proof of Proposition \ref{L2 convergence}]
We note that, in the proof, we can formally take $\kappa=0$ for simplicity of notation.
Indeed, if we let $\tilde{u} \deq \ee^{\ii \kappa t}\,u$, then $\tilde{u}$ solves \eqref{NLS} with $\kappa=0$. Likewise $\tilde{u}^{\epsilon} \deq \ee^{\ii \kappa t}\,u^{\epsilon}$ solves \eqref{NLS_epsilon} with $\kappa=0$. 
Finally, we note that \eqref{L2 convergence claim} is equivalent to showing that for all $T>0$ we have
\begin{equation*}
\lim_{\epsilon \to 0} \|\tilde{u}^{\epsilon}-\tilde{u}\|_{L^{\infty}_{[-T,T]}\fra H} \;=\; 0\,.
\end{equation*}
Throughout the proof, we fix $T>0$ and consider $|t| \leq T$. In what follows, we assume $t \geq 0$. The negative times are treated by an analogous argument.

Before we proceed, we briefly recall the arguments from \cite[Section 2.6]{Soh} (which, in turn, are based on the arguments from \cite{B}) used to construct the local in time solutions to \eqref{NLS} and \eqref{NLS_epsilon} in $H^s$.  Note that in \cite{Soh}, the quintic NLS was considered. The arguments for the cubic NLS are analogous. In what follows, we outline the main idea and refer the interested reader to the aforementioned reference for more details.

We are looking for \emph{global mild solutions} $u$ to \eqref{NLS}--\eqref{NLS_epsilon}, i.e.\ we want $u$ and $u^{\epsilon}$ to solve 
\begin{align}
\label{Mild solution 1}
u(\cdot,t)&\;=\;
\ee^{\ii t \Delta} \phi_0-\ii \,\int_{0}^{t} \dd t' \, \ee^{\ii (t-t')\Delta} \,|u|^2u(t') 
\\
\label{Mild solution 2}
u^{\epsilon}(\cdot,t)&\;=\;
\ee^{\ii t \Delta} \phi_0-\ii \,\int_{0}^{t} \dd t' \, \ee^{\ii (t-t')\Delta} \,\big(w^\epsilon*|u^{\epsilon}|^2\big)\,u^{\epsilon}(t') 
\end{align}
for almost every $t$.
In what follows, we construct solutions of \eqref{Mild solution 1}-\eqref{Mild solution 2} by constructing mild solutions on a sequence of intervals of fixed length depending on the initial data. Putting these solutions together, we get $u$ and $u^{\epsilon}$.

Let $\chi,\psi \in C_c^\infty(\R)$ be functions such that
\begin{equation}
\label{function chi}
\chi(t)\;=\;
\begin{cases}
1 &\mbox{if }|t| \;\leq\;1\\
0 &\mbox{if }|t| \;>\;2
\end{cases}
\end{equation}
and 
\begin{equation}
\label{function psi}
\psi(t)\;=\;
\begin{cases}
1 &\mbox{if }|t| \;\leq\;2\\
0 &\mbox{if }|t| \;>\;4\,.
\end{cases}
\end{equation}
Given $\delta \in (0,1)$, we define
\begin{equation}
\label{chi_psi_rescaled}
\chi_\delta(t) \;\deq\; \chi\bigg(\frac{t}{\delta}\bigg)\,,\quad \psi_\delta(t) \;\deq\; \psi\bigg(\frac{t}{\delta}\bigg)\,.
\end{equation}
Let us fix $\delta \in (0,1)$ small which we determine later. For $t \in [0,T]$ we consider the map
\begin{multline}
\label{map L}
(Lv)(\cdot,t) \;\deq\; \chi_\delta(t)\, \ee^{\ii t \Delta} \phi_0-\ii \,\chi_\delta(t)\,\int_{0}^{t} \dd t' \, \ee^{\ii (t-t')\Delta} \,|v|^2v(t') 
\\
\;=\;
 \chi_\delta(t)\, \ee^{\ii t \Delta} \phi_0-\ii \,\chi_\delta(t)\,\int_{0}^{t} \dd t' \, \ee^{\ii (t-t')\Delta} \,|v_\delta |^2v_\delta(t')\,,
\end{multline}
where we define the operation
\begin{equation}
\label{v_delta}
v_\delta (x,t) \;\deq\; \psi_\delta(t) \, v(x,t)\,.
\end{equation}
In the last equality in \eqref{map L}, we used \eqref{function chi}--\eqref{chi_psi_rescaled}.
Applying Lemma \ref{Xsb space lemma} (ii)-(vii) and arguing as in the proof of \cite[(2.159)]{Soh} it follows that 
\begin{align}
\label{Lv Xsb}
\|Lv\|_{X^{s,b}} \;&\leq\;C_1\delta^{\frac{1-2b}{2}}\,\|\phi_0\|_{H^s}+C \,\delta^{r_0} \,\|v\|_{X^{0,b}}^2\,\|v\|_{X^{s,b}}
\\
\label{Lv X0b}
\|Lv\|_{X^{0,b}} \;&\leq\;C_1\delta^{\frac{1-2b}{2}}\,\|\phi_0\|_{\fra H}+C \,\delta^{r_0} \,\|v\|_{X^{0,b}}^3\,,
\end{align}
where $C_1>0$ is the constant from Lemma \ref{Xsb space lemma} (ii) corresponding to the cutoff in time given by $\chi_\delta$ and 
\begin{equation}
\label{r_0 definition}
r_0\;\deq\;\frac{1-2b}{2}+3\theta_0\;>\;0\,,
\end{equation}
for $\theta_0>0$ given by \eqref{theta_0 definition} below \footnote{Note that the exact value of $r_0$ is not relevant. The main point is that it is positive. This is ensured by taking $b$ sufficiently close to $1/2$.}.
Note that, from Lemma \ref{Xsb space lemma} (ii) we know that $C_1=C_1(\chi)$.

For clarity, we summarize the ideas of the proof of \eqref{Lv X0b}. The proof of \eqref{Lv Xsb} follows similarly using a duality argument and by applying the fractional Leibniz rule. The latter is rigorously justified by observing that the $X^{\sigma,b}$ norms are invariant under taking absolute values in the spacetime Fourier transform. 
For precise details on the latter point, we refer the reader to \cite[(2.147)--(2.153)]{Soh}.

The estimate for the linear term in \eqref{Lv X0b}  follows immediately from Lemma \ref{Xsb space lemma} (ii) with $\sigma=0$.
Note that, when we apply Lemma \ref{Xsb space lemma} (iv) with $\sigma=0$ for the Duhamel term on the right-hand side of \eqref{map L} we have $b-1<\frac{3}{8}$ and hence we can use Lemma \ref{Xsb space lemma} (vi), H\"{o}lder's inequality and Lemma \ref{Xsb space lemma} (v) to deduce that we have to estimate 
\begin{equation*}
\big\||v_\delta|^2\,v_\delta\big\|_{L^{4/3}_{t,x}} \;\leq\; \|v_\delta\|_{L^{4}_{t,x}}^3
\;\leq \;C\, \|v_\delta\|_{X^{0,3/8}}^3\,.
\end{equation*}
We then deduce the estimate for the Duhamel term in \eqref{Lv X0b} by using Lemma \ref{Lv X0b} (vii).

Analogously, with $r_0>0$ as in \eqref{r_0 definition}, we have
\begin{equation}
\label{X^{0,b} difference}
\|Lv_1-Lv_2\|_{X^{0,b}} \;\leq\;C \,\delta^{r_0} \,\big(\|v_1\|_{X^{0,b}}^2+\|v_2\|_{X^{0,b}}^2\big)\,\|v_1-v_2\|_{X^{0,b}}\,.
\end{equation}
In particular, it follows from \eqref{Lv Xsb}-\eqref{X^{0,b} difference} that $L$ is a contraction on $(\Gamma,\|\cdot\|_{X^{0,b}})$, for 
\begin{equation}
\label{set Gamma}
\Gamma \;\deq\; \Big\{v,\|v\|_{X^{s,b}} \;\leq\;2C_1\delta^{\frac{1-2b}{2}}\,\|\phi_0\|_{H^s}\,,\|v\|_{X^{0,b}} \;\leq\;2C_1\delta^{\frac{1-2b}{2}}\,\|\phi_0\|_{\fra H}\Big\}\,,
\end{equation}
where $\delta \in (0,1)$ is chosen to be \emph{sufficiently small depending on} $\|\phi_0\|_{\fra H}$. By arguing as in the proof of \cite[Proposition 2.3.2]{Soh} (whose proof, in turn, is based on that of \cite[Theorem 1.2.5]{Cazenave}), it follows that $(\Gamma,\|\cdot\|_{X^{0,b}})$ is a Banach space. Therefore, we obtain a unique fixed point of $L$ in $\Gamma$. We refer the reader to \cite[Section 2.5]{Soh} for more details.

Moreover, suppose that for some other $\tilde{\delta}>0$ the function $\tilde{v} \in X^{0,b}$ solves
\begin{equation*}
\tilde{v} \;=\; \chi_{\tilde{\delta}}(t)\, \ee^{\ii t \Delta} \phi_0-\ii \,\chi_{\tilde{\delta}}(t)\,\int_{0}^{t} \dd t' \, \ee^{\ii (t-t')\Delta} \,|\tilde{v}|^2\tilde{v}(t')\,.
\end{equation*}
We then want to argue that 
\begin{equation}
\label{uniqueness}
v\big|_{\Lambda \times [0,\hat{\delta}]}\;=\;\tilde{v}\big|_{\Lambda \times [0,\hat{\delta}]}\quad \mbox {for all} \quad \hat{\delta} \in \big[0,\min\{\delta,\tilde{\delta}\}\big]\,.
\end{equation}
In order to prove \eqref{uniqueness}, we need to work in \emph{local} $X^{\sigma,b}$ \emph{spaces}. Given $\sigma \in \R$ and a time interval $I$, we define
\begin{equation*}
\|f\|_{X^{\sigma,b}_I} \;\deq\; \inf \Big\{\|g\|_{X^{\sigma,b}}\,,\,g\big|_{\Lambda \times I}\;=\;f\big|_{\Lambda \times I} \Big\}\,.
\end{equation*}
In particular, we have that $v,\tilde{v} \in X^{0,b}_{[0,\hat{\delta}]}$.
Noting that for $t \in [0,\hat{\delta}]$ we have $\chi_{\delta}(t)=\chi_{\tilde{\delta}}(t)$, the same arguments used to show \eqref{X^{0,b} difference} imply that
\begin{equation*}
\|v-\tilde{v}\|_{X^{0,b}_{[0,\hat{\delta}]}} \;\leq\; C\,\hat{\delta}^{\,r_0} \, \Big(\|v\|_{X^{0,b}_{[0,\hat{\delta}]}}+\|\tilde{v}\|_{X^{0,b}_{[0,\hat{\delta}]}}\Big)^2\,\|v-\tilde{v}\|_{X^{0,b}_{[0,\hat{\delta}]}}
\;\leq\; C\,\hat{\delta}^{\,r_0} \, \Big(\|v\|_{X^{0,b}}+\|\tilde{v}\|_{X^{0,b}}\Big)^2\,\|v-\tilde{v}\|_{X^{0,b}_{[0,\hat{\delta}]}}
\,.
\end{equation*}
Here $r_0>0$ is given by \eqref{r_0 definition}. 
We hence deduce \eqref{uniqueness} for sufficiently small $\hat{\delta}$. By an additional iteration argument, we deduce \eqref{uniqueness} for the full range $\hat{\delta} \in \big[0,\min\{\delta,\tilde{\delta}\}\big]$.

Likewise, for $\epsilon>0$, we consider the map
\begin{equation}
\label{map L_epsilon}
(L^{\epsilon}v)(\cdot,t) \;\deq\; \chi_\delta(t)\, \ee^{\ii t \Delta} \phi_0-\ii \,\chi_\delta(t)\,\int_{0}^{t} \dd t' \, \ee^{\ii (t-t')\Delta} \,\big(w^\epsilon*|v_\delta|^2\big)v_\delta(t')\,,
\end{equation}
for $v_\delta$ given as in \eqref{v_delta}.
We note that, for all $k \in \N$, we have 
\begin{equation*}
\widehat{w^\epsilon}(k)= \int_{-\frac{1}{2}}^{\frac{1}{2}} \frac{1}{\epsilon} \,w\bigg(\frac{x}{\epsilon}\bigg) \,\ee^{-2\pi \ii k x} \,\dd x
= \int_{-\frac{1}{2\epsilon}}^{\frac{1}{2\epsilon}} w(y) \,\ee^{-2\pi \ii \epsilon k y} \,\dd y\,.
\end{equation*}
Therefore, by the assumptions on $w$, it follows that 
\begin{equation}
\label{w_epsilon^hat}
|\widehat{w^\epsilon}(k)| \;\leq\;C
\end{equation}
for some $C>0$ independent of $k,\epsilon$. 
Using the same arguments as for $L$ and applying \eqref{w_epsilon^hat}, it follows that $L^{\epsilon}$ defined in \eqref{map L_epsilon} has a unique fixed point $v^\epsilon  \in \Gamma$. Moreover, a statement analogous to \eqref{uniqueness} holds.

We then define $u,u^{\epsilon}$ on $[0,\delta]$ according to 
\begin{equation}
\label{Definition of u 1}
u\big|_{\Lambda \times [0,\delta]} \;\deq\; v \big|_{\Lambda \times [0,\delta]}\,,\quad u^{\epsilon}\big|_{\Lambda \times [0,\delta]} \;\deq\; v^\epsilon \big|_{\Lambda \times [0,\delta]}\,.
\end{equation}

We now iterate this construction. In doing so, we note that the increment $\delta \in (0,1)$ we chose above \emph{depends only on the $L^2$ norm of the initial data} and hence is the same at every step of the iteration. More precisely, for all $n \in \N$ with $n \leq (T-1)/\delta$, we construct $v_{(n)}$, and $v_{(n)}^{\epsilon}$ such that the following properties hold.
\begin{itemize}
\item[(i)] $v_{(n)}$ is a mild solution of the local NLS \eqref{localNLS1} with $\kappa=0$ in the sense of \eqref{Mild solution 1} on the time interval $\big[n\delta,(n+1)\delta\big]$ and we have
\begin{equation}
\label{v^n Xsb}
\|v_{(n)}\|_{X^{s,b}} \;\leq\; 2C_1\, \delta^{\frac{1-2b}{2}}\, \|v_{(n)}(n\delta)\|_{H^s}\,.
\end{equation}
\item[(ii)] $v_{(n)}^\epsilon$ is a mild solution of \eqref{NLS_epsilon} with $\kappa=0$ in the sense of \eqref{Mild solution 2} on the time interval $\big[n\delta,(n+1)\delta\big]$ and we have
\begin{equation}
\label{v^n_epsilon Xsb}
\|v_{(n)}^\epsilon\|_{X^{s,b}} \;\leq\; 2C_1\,\delta^{\frac{1-2b}{2}}\, \|v_{(n)}^\epsilon(n\delta)\|_{H^s}\,.
\end{equation}
\end{itemize}
We then generalize \eqref{Definition of u 1} by defining $u,u^{\epsilon}$ on $[n\delta,(n+1)\delta]$ according to 
\begin{equation}
\label{Definition of u}
u\big|_{\Lambda \times [n\delta,(n+1)\delta]} \;\deq\; v_{(n)} \big|_{\Lambda \times [n\delta,(n+1)\delta]}\,,\quad u^{\epsilon}\big|_{\Lambda \times [n\delta,(n+1)\delta]} \;\deq\; v_{(n)}^{\epsilon}\big|_{\Lambda \times [n\delta,(n+1)\delta]}\,.
\end{equation}
Note that, in this definition, $v_{(0)}=v$ and $v_{(0)}^{\epsilon}=v^\epsilon $.

We observe that by \eqref{uniqueness} and the analogous uniqueness statement for $L^\epsilon$, this construction does not depend on $\delta$ (as long as $\delta$ is chosen to be small enough, depending in $\|\phi_0\|_{\fra H}$, c.f.\  \eqref{Lv Xsb}--\eqref{set Gamma}). In particular, we can choose $\delta=\delta_0(\|\phi_0\|_{\fra H})$.
By \eqref{Definition of u}, Lemma \ref{Xsb space lemma} (i), \eqref{v^n Xsb}--\eqref{v^n_epsilon Xsb}, it follows that 
\begin{equation}
\label{H^s bound}
\begin{cases}
\|u\|_{L^\infty_{[0,T]} H^s_x}\;\leq\; C(\|\phi_0\|_{H^s},T)\\
\|u^{\epsilon}\|_{L^\infty_{[0,T]} H^s_x} \;\leq\; C(\|\phi_0\|_{H^s},T)\,,
\end{cases}
\end{equation}
for some finite quantity $C(\|\phi_0\|_{H^s},T)>0$. It is important to note that $C(\|\phi_0\|_{H^s},T)$ is independent of $\delta$. In other words, if we choose $\delta$ smaller, then the same bounds in \eqref{H^s bound} hold.

Using \eqref{Definition of u 1} and Lemma \ref{Xsb space lemma} (i), it follows that, for $\delta \in (0,1)$ chosen sufficiently small as earlier, we have
\begin{equation*}
\|u-u^{\epsilon}\|_{L^\infty_{[0,\delta]}L^2_x} \;=\;\|v-v^\epsilon \|_{L^\infty_{[0,\delta]}L^2_x} \;\leq\; C\,\|v-v^\epsilon \|_{X^{0,b}}\,.
\end{equation*}
By construction of $v$ and $v^\epsilon $, we obtain
\begin{multline*}
\|v-v^\epsilon\|_{X^{0,b}} \;\leq\; \bigg\|\chi_{\delta}(t)\,\int_{0}^{t} \dd t' \, \ee^{\ii (t-t')\Delta} \,\big(|v_\delta (t')|^2-w^\epsilon*|v_\delta (t')|^2\big)\,v_\delta(t')\bigg\|_{X^{0,b}}
\\
+\bigg\|\chi_{\delta}(t)\,\int_{0}^{t} \dd t' \, \ee^{\ii (t-t')\Delta}\,\Big\{ w^\epsilon *\,\big(|v_\delta(t')|^2-|v_\delta^\epsilon(t')|^2\big)\Big\} \,v_\delta(t')\bigg\|_{X^{0,b}}
\\
+ \bigg\|\chi_{\delta}(t)\,\int_{0}^{t} \dd t' \, \ee^{\ii (t-t')\Delta} \,\Big\{w^\epsilon * |v_\delta^\epsilon(t')|^2\Big\}\,\big(v_\delta(t')-v_\delta^\epsilon(t')\big)\bigg\|_{X^{0,b}}\,,
\end{multline*}
which by Lemma \ref{Xsb space lemma} (iv) is
\begin{multline}
\label{v-v_epsilon bound}
\;\leq\; C\,\delta^{\frac{1-2b}{2}}\,\Big\|\big(|v_\delta|^2-w^\epsilon*|v_\delta|^2\big)\,v_\delta\Big\|_{X^{0,b-1}}
+C\,\delta^{\frac{1-2b}{2}}\,\Big\|\Big\{w^\epsilon *\,\big(|v_\delta|^2-|v_\delta^\epsilon|^2\big)\Big\}\,v_\delta\Big\|_{X^{0,b-1}} 
\\
+ C\,\delta^{\frac{1-2b}{2}}\,\Big\|\Big\{w^\epsilon * |v_\delta^\epsilon|^2\Big\}\,\big(v_\delta-v_\delta^\epsilon\big)\Big\|_{X^{0,b-1}}\,.
\end{multline}
Note that, in the above expressions, the quantity $v_\delta^\epsilon$ is obtained from $v^\epsilon$ according to \eqref{v_delta}.
We now estimate each of the terms on the right-hand side of \eqref{v-v_epsilon bound} separately.

For the first term, we note that for fixed $x \in \Lambda$ we have
\begin{equation}
\label{triangle_inequality_bound}
\big||v_\delta(x)|^2-(w^\epsilon * |v_\delta|^2)(x)\big| \;\leq\; \int \dd y\, w^\epsilon(x-y)\,|v_\delta(x)-v_\delta(y)| \, \big(|v_\delta(x)|+|v_\delta(y)|\big)\,.
\end{equation}
Here, we used \eqref{w_epsilon} by which we obtain that
\begin{equation}
\label{integral w_epsilon}
\int \dd x\, w^\epsilon(x) \;=\;1\,.
\end{equation}
Moreover, we used the elementary inequality 
\begin{equation}
\label{difference of squares}
\big| |a_1|^2-|a_2|^2\big| \leq |a_1-a_2| \big(|a_1|+|a_2|\big)\,.
\end{equation}
We recall \eqref{definition of b} and use Lemma \ref{Xsb space lemma} (vi) to note that 
\begin{multline}
\label{v-v_epsilon 1 Step 1}
\Big\|\big(|v_\delta|^2-w^\epsilon*|v_\delta|^2\big)\,v_\delta\Big\|_{X^{0,b-1}} \;\leq\; \Big\|\big(|v_\delta|^2-w^\epsilon*|v_\delta|^2\big)\,v_\delta\Big\|_{X^{0,-3/8}}
\\
\;\leq\; C\,\Big\|\big(|v_\delta|^2-w^\epsilon*|v_\delta|^2\big)\,v_\delta\Big\|_{L^{4/3}_{t,x}} \,,
\end{multline}
which by \eqref{triangle_inequality_bound} is
\begin{equation}
\label{v-v_epsilon 1 Step 1 part 2}
\;\leq\; C\,\Big\| w^\epsilon(x-y)\,|v_\delta(x)-v_\delta(y)| \, |v_\delta(x)|\Big\|_{L^{4/3}_tL^{4/3}_xL^1_y}
+ C\,\Big\| w^\epsilon(x-y)\,|v_\delta(x)-v_\delta(y)| \, |v_\delta(y)|\Big\|_{L^{4/3}_tL^{4/3}_xL^1_y}\,.
\end{equation}
Note that, in order to apply \eqref{triangle_inequality_bound} in \eqref{v-v_epsilon 1 Step 1}, it is crucial to use that we are estimating the $L^{4/3}_{t,x}$ norm and not the $X^{0,b-1}$ norm. 

By H\"{o}lder's inequality in mixed-norm spaces and by the construction of $v_\delta$ in \eqref{v_delta}, the expression in \eqref{v-v_epsilon 1 Step 1 part 2} is
\begin{multline}
\label{v-v_epsilon 1*}
\;\leq\;\big\|[x-y]^{s+\frac{1}{2}}\,w^\epsilon(x-y)\big\|_{L^{\infty}_xL^2_y} \,\big\|\psi_\delta\big\|_{L^\infty_t} \,  \, 
\Bigg\|\frac{v(x)-v(y)}{[x-y]^{s+\frac{1}{2}}}\Bigg\|_{L^{\infty}_tL^2_{x,y}}\,\big\|\psi_\delta\big\|_{L^{4/3}_t} \, \big\|v(x)\big\|_{L^{\infty}_tL^4_x}
\\
+ \big\|[x-y]^{s+\frac{1}{2}}\,w^\epsilon(x-y)\big\|_{L^{4}_{x,y}} \,\big\|\psi_\delta\big\|_{L^\infty_t} \, \Bigg\|\frac{v(x)-v(y)}{[x-y]^{s+\frac{1}{2}}}\Bigg\|_{L^{\infty}_tL^2_{x,y}}\,\big\|\psi_\delta\big\|_{L^{4/3}_t} \,\big\|v(y)\big\|_{L^{\infty}_tL^4_y}\,.
\end{multline}

We now estimate \eqref{v-v_epsilon 1*}.
By \eqref{w_epsilon} and by the assumption that $w \in C_c^{\infty}(\R)$, we have that for all $1 \leq p <\infty$
\begin{multline}
\label{v-v_epsilon 1A}
\big\|[z]^{s+\frac{1}{2}}\,w^\epsilon(z)\big\|_{L^p_{z}}
\;=\;\frac{1}{\epsilon}\, \bigg(\int_0^1\,\dd z \,z^{(s+\frac{1}{2})p}\, \Big|w\Big(\frac{z}{\epsilon}\Big)\Big|^p \bigg)^{1/p}
\;=\;
\epsilon^{s+\frac{1}{p}-\frac{1}{2}} \, \bigg(\int_0^{1/\epsilon} \dd \tilde{z} \,\tilde{z}^{(s+\frac{1}{2})p}\,|w(\tilde{z})|^p \bigg)^{1/p}
\\
\;\leq\; C(s,p,w) \,\epsilon^{s+\frac{1}{p}-\frac{1}{2}}\,.
\end{multline}
Here we used the change of variables $\tilde{z}=z/\epsilon$.
We now apply \eqref{v-v_epsilon 1A} with $p=2$ and $p=4$ and deduce that
\begin{equation}
\label{v-v_epsilon 1A application}
\big\|[x-y]^{s+\frac{1}{2}}\,w^\epsilon(x-y)\big\|_{L^{\infty}_xL^2_y} \;\leq\; C \epsilon^s\,,\quad \big\|[x-y]^{s+\frac{1}{2}}\,w^\epsilon(x-y)\big\|_{L^{4}_{x,y}} \;\leq\;C \epsilon^{s-\frac{1}{4}}\,.
\end{equation}
For the second inequality in \eqref{v-v_epsilon 1A application}, we also used the compactness of $\Lambda$.
Furthermore, by Lemma \ref{Sobolev space torus}, Lemma \ref{Xsb space lemma} (i) and since $v \in \Gamma$ for the set $\Gamma$ defined as in \eqref{set Gamma}, it follows that
\begin{equation}
\label{v-v_epsilon 1B}
\Bigg\|\frac{v(x)-v(y)}{[x-y]^{s+\frac{1}{2}}}\Bigg\|_{L^{\infty}_tL^2_{x,y}} \;\leq\; C\, \|v\|_{L^{\infty}_tH^s_x} \;\leq\; C\, \|v\|_{X^{s,b}} \;\leq\;C\delta^{\frac{1-2b}{2}}\,\|\phi_0\|_{H^s}\,.
\end{equation}
Moreover, we note that, by H\"{o}lder's inequality, Sobolev embedding with $s \geq \frac{1}{4}$ and the same arguments as in \eqref{v-v_epsilon 1B}, we have
\begin{equation}
\label{v-v_epsilon 1C}
\|v\|_{L^\infty_t L^4_x} \;\leq\;C \,\|v\|_{L^\infty_tH^s_x}\;\leq\,C\delta^{\frac{1-2b}{2}}\,\|\phi_0\|_{H^s}\,.
\end{equation}
We use \eqref{v-v_epsilon 1A application}--\eqref{v-v_epsilon 1C}, as well as \eqref{function psi}--\eqref{chi_psi_rescaled} to deduce that the expression in \eqref{v-v_epsilon 1*} is 
\begin{equation}
\label{v-v_epsilon 1}
\;\leq\; C\,\delta^{\frac{3}{4}+(1-2b)}\, \|\phi_0\|_{H^s}^2 \, \epsilon^{s-1/4}\,.
\end{equation}

We now estimate the second term on the right-hand of \eqref{v-v_epsilon bound}. By applying Lemma \ref{Xsb space lemma} (vi) as in \eqref{v-v_epsilon 1 Step 1}, it follows that
\begin{multline*}
\Big\|\Big\{w^\epsilon *\,\big(|v_\delta|^2-|v_\delta^\epsilon|^2\big)\Big\}\,v_\delta\Big\|_{X^{0,b-1}} \;\leq\;C\, \Big\|\Big\{w^\epsilon *\,\big(|v_\delta|^2-|v_\delta^\epsilon|^2\big)\Big\}\,v_\delta\Big\|_{L^{4/3}_{t,x}}
\\
\;\leq\;C\, \Big\|\Big\{w^\epsilon *\,\Big(|v_\delta-v_\delta^\epsilon| \big(|v_\delta|+|v_\delta^\epsilon \big)\Big)\Big\}\,v_\delta\Big\|_{L^{4/3}_{t,x}}
\,. 
\end{multline*}
In the last inequality, we also used \eqref{difference of squares}. By H\"{o}lder's and Young's inequality, it follows that this expression is
\begin{multline}
\label{v-v_epsilon 2}
\;\leq\;\big\|v_\delta-v_\delta^\epsilon\big\|_{L^2_{t,x}} \, \Big(\big\|v_\delta\big\|_{L^8_{t,x}}+\big\|v_\delta^\epsilon\big\|_{L^8_{t,x}} \Big) \,\big\|v_\delta\big\|_{L^8_{t,x}}
\\
\;\leq\; C\,\delta^{\frac{3}{4}}\,
\|v-v^\epsilon\|_{L^\infty_t L^2_x} \, \big(\|v\|_{L^\infty_t L^8_x}+\|v^\epsilon\|_{L^\infty_t L^8_x} \big) \,\|v\|_{L^\infty_t L^8_x}
\\
\;\leq\; C\,\delta^{\frac{3}{4}}\,
\|v-v^\epsilon\|_{X^{0,b}} \, \big(\|v\|_{X^{s,b}}+\|v^\epsilon\|_{X^{s,b}} \big) \,\|v\|_{X^{s,b}} \;\leq\; C\, 
{\delta}^{\frac{3}{4}+(1-2b)}\,\|\phi_0\|_{H^s}^2\,\|v-v^\epsilon\|_{X^{0,b}}\,.
\end{multline}
Above we used Sobolev embedding with $s \geq \frac{3}{8}$, Lemma \ref{Xsb space lemma} (i), the construction of $v,v_\delta,v_\delta^\epsilon,v,v^\epsilon$, as well as $\|w^{\epsilon}\|_{L^1}=1$, which follows from \eqref{integral w_epsilon} since $w^{\epsilon} \geq 0$.

The third term on the right-hand side of \eqref{v-v_epsilon bound} is estimated in a similar way. Arguing as in \eqref{v-v_epsilon 1 Step 1}, we need to estimate
\begin{multline}
\label{v-v_epsilon 3}
\Big\|\Big\{w^\epsilon * |v_\delta^\epsilon|^2\Big\}\,\big(v_\delta-v_\delta^\epsilon\big)\Big\|_{L^{4/3}_{t,x}} \;\leq\; \big\|v_\delta^\epsilon\big\|_{L^8_{t,x}}^2\,\big\|v_\delta-v_\delta^\epsilon\big\|_{L^2_{t,x}} \;\leq\; C\, \delta^{\frac{3}{4}} \,\|v^\epsilon\|_{L^\infty_t L^8_x}^2 \,\|v-v^\epsilon\|_{L^\infty_tL^2_x}
\\
\;\leq\;
C\, \delta^{\frac{3}{4}} \,\|v^\epsilon\|_{X^{s,b}}^2 \,\|v-v^\epsilon\|_{X^{0,b}}  \;\leq\; C\, 
{\delta}^{\frac{3}{4}+(1-2b)}\,\|\phi_0\|_{H^s}^2\,\|v-v^\epsilon\|_{X^{0,b}}\,.
\end{multline}
Here, we again used H\"{o}lder's inequality, Young's inequality, Sobolev embedding with $s \geq \frac{3}{8}$, Lemma \ref{Xsb space lemma} (i) and the construction of $v_\delta,v_\delta^\epsilon,v^\epsilon,w^{\epsilon}$.
Substituting \eqref{v-v_epsilon 1}--\eqref{v-v_epsilon 3} into \eqref{v-v_epsilon bound}, it follows that
\begin{equation}
\label{v-v_epsilon bound 2}
\|v-v^\epsilon\|_{X^{0,b}} \;\leq\; C\,\delta^{\theta_1}\, \|\phi_0\|_{H^s}^2 \, \epsilon^{s-\frac{1}{4}}+C\, 
{\delta}^{\theta_1}\,\|\phi_0\|_{H^s}^2\,\|v-v^\epsilon\|_{X^{0,b}}\,,
\end{equation}
where 
\begin{equation*}
\theta_1 \;\deq\; \frac{3}{4}+\frac{3(1-2b)}{2}>0\,.
\end{equation*}
In particular, if we choose $\delta \equiv \delta(\|\phi_0\|_{H^s})>0$ possibly smaller than before so that the coefficient of $\|v-v^\epsilon\|_{X^{0,b}}$ on the right-hand side of \eqref{v-v_epsilon bound 2} is smaller than $1/2$, it follows that 
\begin{equation}
\label{v-v_epsilon bound 3}
\|v-v^\epsilon\|_{X^{0,b}} \;\leq\; C(\|\phi_0\|_{H^s})\,\epsilon^{s-\frac{1}{4}}\,.
\end{equation}
By analogous arguments, we obtain more generally that for all $n \in \N$ we have
\begin{multline}
\label{v_{(n)}-v_{(n)}_epsilon 1}
\|v_{(n)}-v_{(n)}^{\epsilon}\|_{X^{0,b}} \;\leq\; 
C\delta^{\frac{1-2b}{2}} \, \|v_{(n)}(n\delta)-v_{(n)}^{\epsilon}(n\delta)\|_{\fra H}
+
C\,\delta^{\theta_1}\, \|v_{(n)}(n\delta)\|_{H^s}^2 \, \epsilon^{s-\frac{1}{4}}
\\
+C\, 
{\delta}^{\theta_1}\,\big(\|v_{(n)}(n\delta)\|_{H^s} + \|v_{(n)}^\epsilon(n\delta)\|_{H^s} \big)^2\,\|v_{(n)}-v_{(n)}^\epsilon\|_{X^{0,b}}\,.
\end{multline}
Note that the first term on the right-hand side of \eqref{v_{(n)}-v_{(n)}_epsilon 1} appears because in general we consider different initial data $v_{(n)}(n\delta)$ and $v_{(n)}^{\epsilon}(n\delta)$.  We hence obtain the corresponding term on the right-hand side of \eqref{v_{(n)}-v_{(n)}_epsilon 1} by Lemma \ref{Xsb space lemma} (ii).

In particular, if $1 \leq n \leq (T-1)/\delta$, we obtain by Lemma \ref{Xsb space lemma} (i), \eqref{Definition of u}--\eqref{H^s bound} and \eqref{v_{(n)}-v_{(n)}_epsilon 1} that
\begin{multline*}
\|v_{(n)}-v_{(n)}^{\epsilon}\|_{X^{0,b}} \;\leq\; C\,\delta^{\frac{1-2b}{2}}\,\|v_{(n-1)}-v_{(n-1)}^{\epsilon}\|_{X^{0,b}}
\\
+C_1(\|\phi_0\|_{H^s},T)\,\epsilon^{s-\frac{1}{4}}+C_2(\|\phi_0\|_{H^s},T)\,\delta^{\theta_1}\,\|v_{(n)}-v_{(n)}^{\epsilon}\|_{X^{0,b}}\,.
\end{multline*}
Here we also assume that $\delta<1$. In particular, choosing $\delta\equiv \delta(\|\phi_0\|_{H^s},T)>0$ even smaller than before such that $C_2(\|\phi_0\|_{H^s},T)\,\delta^{\theta_1}<1/2$, it follows that 
\begin{equation}
\label{v_{(n)}-v_{(n)}_epsilon 2}
\|v_{(n)}-v_{(n)}^{\epsilon}\|_{X^{0,b}} \;\leq\; C(\|\phi_0\|_{H^s},T)\,\|v^{(n-1)}-v^{(n-1)}_\epsilon\|_{X^{0,b}}+C(\|\phi_0\|_{H^s},T)\,\epsilon^{s-\frac{1}{4}}\,,
\end{equation}
for all $1 \leq n \leq (T-1)/\delta$. We note that, by \eqref{H^s bound}, we can take
\begin{equation*}
\delta \;\equiv\; \delta \Bigg(\sup_{[0,T]} \|u(t)\|_{H^s}+\sup_{\epsilon >0} \sup_{[0,T]} \|u^{\epsilon}(t)\|_{H^s}\Bigg) \;=\; \delta(\|\phi_0\|_{H^s},T) \;>\;0\,
\end{equation*}
in \eqref{v_{(n)}-v_{(n)}_epsilon 2}.

Iterating \eqref{v_{(n)}-v_{(n)}_epsilon 2} and recalling \eqref{v-v_epsilon bound 3}, it follows that for all $0 \leq n \leq (T-1)/\delta$ we have
\begin{equation}
\label{v_{(n)}-v_{(n)}_epsilon conclusion}
\|v_{(n)}-v_{(n)}^{\epsilon}\|_{X^{0,b}}\;\leq\;C(\|\phi_0\|_{H^s},T)\,\epsilon^{s-\frac{1}{4}}\,.
\end{equation}
Using Lemma \ref{Xsb space lemma} (i), \eqref{Definition of u} and \eqref{v_{(n)}-v_{(n)}_epsilon conclusion}, it follows that 
\begin{equation*}
\|u-u^{\epsilon}\|_{L^{\infty}_{[0,T]}\fra H} \;\leq\; C(\|\phi_0\|_{H^s},T)\,\epsilon^{s-\frac{1}{4}}\,,
\end{equation*}
from where we deduce the claim since $s \geq \frac{3}{8}$.
\end{proof}

Before proceeding to the proof of Theorem \ref{Main Result_local_nonlinearity} we record the following elementary lemma.

\begin{lemma}
\label{diagonal_argument_zeta}
Let $(Z_k)_{k \in \N}$ be an increasing family of sets (i.e.\ $Z_k \subset Z_{k+1}$), and set $Z \deq \bigcup_{k \in \N}Z_k$.
For $\epsilon,\tau > 0$ let $f,f^\epsilon,f^\epsilon_\tau : Z \rightarrow \C$ be functions which satisfy the following properties.
\begin{enumerate}
\item
For each fixed $k \in \N$ and $\epsilon > 0$ we have 
$
\lim_{\tau \to \infty} f_\tau^\epsilon(\zeta) = f^\epsilon(\zeta)
$
uniformly in $\zeta \in Z_k$.
\item
For each fixed $k \in \N$ we have 
$\lim_{\epsilon \to 0} f^\epsilon(\zeta) = f(\zeta)$
uniformly in $\zeta \in Z_k$.
\end{enumerate}
Then there exists a sequence of positive numbers $(\epsilon_\tau)$, with $\lim_{\tau \to \infty} \epsilon_\tau = 0$, such that
\begin{equation*}
\lim_{\tau \to \infty} f_\tau^{\epsilon_\tau}(\zeta) \;=\; f(\zeta)\,,
\end{equation*}
for all $\zeta \in Z$.
\end{lemma} 

\begin{proof} Assumptions (i) and (ii) combined with a diagonal argument imply that, for a fixed $k \in \N$, there exists a sequence of positive numbers $(\epsilon^k_\tau)$, with $\lim_{\tau \to \infty} \epsilon^k_\tau = 0$, such that $\lim_{\tau \to \infty} f_\tau^{\epsilon^k_\tau}(\zeta) = f(\zeta)$, uniformly in $\zeta \in Z_k$. Using a further diagonal argument, we extract a diagonal sequence $(\epsilon_\tau)$ from 
$(\epsilon^k_\tau)$ such that
$\lim_{\tau \to \infty} f_\tau^{\epsilon_\tau}(\zeta) = f(\zeta),$
for all $\zeta \in Z$.
\end{proof}

\begin{proof}[Proof of Theorem \ref{Main Result_local_nonlinearity}]
We shall apply Lemma \ref{diagonal_argument_zeta} for the following choices of the sets $Z_{k}$ and $Z$
\begin{align*}
Z &\;\deq\; \Big\{(m,t_1,\ldots,t_m,p_1,\ldots,p_m,\xi^1,\ldots,\xi^m)\,:\, m \in \N, \, t_i \in \R, \, p_i \in \N, \, \xi^i \in \cal L (\fra H^{(p_i)})\Big\}
\\
Z_k &\;\deq\; \Big\{(m,t_1,\ldots,t_m,p_1,\ldots,p_m,\xi^1,\ldots,\xi^m) \in Z \,:\, m \leq k, \, \abs{t_i} \leq k, \, p_i \leq k, \, \norm{\xi^i} \leq k\Big\}
\end{align*}
and the functions
\begin{equation*}
f^\epsilon_\sharp(\zeta) \;\deq\; \rho^\epsilon_\sharp \pb{\Psi^{t_1,\epsilon}_{\sharp} \,\Theta_{\sharp}(\xi^1)\,\cdots \, 
\Psi^{t_m,\epsilon}_{\sharp} \,\Theta_{\sharp}(\xi^m)}\,, \qquad
f(\zeta) \;\deq\; \rho \pb{\Psi^{t_1} \,\Theta(\xi^1)\,\cdots \, 
\Psi^{t_m} \,\Theta(\xi^m)}\,,
\end{equation*}
where $\sharp$ stands for either nothing or $\tau$.

By Theorem \ref{Main Result}, here used for the choice $w=w^\epsilon$, Lemma \ref{diagonal_argument_zeta}, and Remark \ref{uniformity_remark}, it suffices to show that, for fixed $k \in \N$, 
\begin{equation}
\label{local_nonlinearity_claim}
\lim_{\epsilon \to 0} \rho^\epsilon \Big(\Psi^{t_1,\epsilon} \,\Theta(\xi^1)\,\cdots \, 
\Psi^{t_m,\epsilon} \,\Theta(\xi^m)\Big)  \;=\;
\rho \Big(\Psi^{t_1} \,\Theta(\xi^1)\,\cdots \, 
\Psi^{t_m} \,\Theta(\xi^m)\Big)\,,
\end{equation}
uniformly in the parameters
\begin{equation}
\label{uniformity_k_epsilon}
m \;\leq\; k\,, \,\abs{t_i} \;\leq\; k\,, \, p_i \;\leq\; k\,,\,\|\xi^i\| \;\leq\;k\,,\, i=1,\dots, m.
\end{equation}

In the sequel, we set  $w=w^\epsilon$ and we add a superscript $\epsilon$ to any quantity defined in terms of $w$ to indicate that in its definition $w$ is replaced by $w^{\epsilon}$. For instance, we write the expectation $\rho^\epsilon(X)$ as in \eqref{rho_frac}, the classical interaction $\cal W^\epsilon$ 
defined as in \eqref{classical interaction} and $\Psi^{t,\epsilon} \,\Theta(\xi)$ given as in Definition \ref{Classical_time_evolution}, all with this modification. 
In addition, we write the classical interaction $\cal W$ defined as in \eqref{classical interaction} with $w$ formally set to equal the delta function. With these conventions we define
\begin{equation*} 
\tilde \rho_z^\epsilon(X) \;\deq\; \int X \, \ee^{- z \cal W^\epsilon} \, \dd \mu\,,\quad \tilde \rho_z(X) \;\deq\; \int X \, \ee^{- z \cal W} \, \dd \mu
\end{equation*}
for a random variable $X$ and $\re z \geq 0$. In particular, we have
\begin{equation} \label{rho_quotient_1}
\rho^\epsilon(X)\;=\;\frac{\tilde \rho_1^\epsilon(X)}{\tilde \rho_1^\epsilon(1)}\,,\quad \rho(X)\;=\;\frac{\tilde \rho_1(X)}{\tilde \rho_1(1)}\,.
\end{equation}
Let us first observe that 
\begin{equation}
\label{W convergence}
\lim_{\epsilon \to 0} \cal W^\epsilon \;=\; \cal W \quad \mbox{almost surely.}
\end{equation}
Indeed, using \eqref{integral w_epsilon}--\eqref{difference of squares}, we obtain
\begin{equation*}
\big|2\,(\cal W-\cal W^\epsilon)\big|
\;\leq\; \int \dd x\, \dd y\, w^\epsilon(x-y) \, \big|\phi(x)-\phi(y)\big| \, \big(|\phi(x)|+|\phi(y)|\big)\,|\phi(x)|^2\,.
\end{equation*}
Set $s = \frac{3}{8}$. Note that, since $s<\frac{1}{2}$, the free classical field $\phi$ defined in \eqref{classical_free_field}  is in $H^s(\Lambda)$ almost surely.
We now apply H\"{o}lder's inequality in mixed norm spaces similarly as in \eqref{v-v_epsilon 1*} to deduce that this expression is
\begin{multline}
\label{W convergence 1}
\;\leq\; \big\|[x-y]^{s+\frac{1}{2}}\,w^\epsilon(x-y)\big\|_{L^{\infty}_xL^2_y} \, 
\Bigg\|\frac{\phi(x)-\phi(y)}{[x-y]^{s+\frac{1}{2}}}\Bigg\|_{L^2_{x,y}}\,\|\phi(x)\|_{L^6_x}^3
\\
+ \big\|[x-y]^{s+\frac{1}{2}}\,w^\epsilon(x-y)\big\|_{L^{\infty}_xL^4_y} \, 
\Bigg\|\frac{\phi(x)-\phi(y)}{[x-y]^{s+\frac{1}{2}}}\Bigg\|_{L^2_{x,y}}\,\|\phi(y)\|_{L^4_y} \,\|\phi(x)\|_{L^4_x}^2\;\leq\; C\,\epsilon^{s-\frac{1}{4}} \, \|\phi\|_{H^s}^4\,.
\end{multline}
Here we used \eqref{v-v_epsilon 1A application}, Lemma \ref{Sobolev space torus} and Sobolev embedding with $s \geq \frac{1}{4}$.
The claim \eqref{W convergence} now follows from \eqref{W convergence 1} since $\phi \in H^s(\Lambda)$ almost surely.

Since $\cal W^\epsilon,\cal W \geq 0$, it follows from \eqref{W convergence} and the dominated convergence theorem that
\begin{equation} \label{rho_quotient_2}
\lim_{\epsilon \to 0} \tilde{\rho}_1^\epsilon(1) \;=\; \tilde{\rho}_1(1)\,.
\end{equation}
In particular, by \eqref{rho_quotient_1} and \eqref{rho_quotient_2} we deduce that \eqref{local_nonlinearity_claim} is equivalent to showing that
\begin{equation}
\label{local_nonlinearity_claim2}
\lim_{\epsilon \to 0}\tilde{\rho}_1^\epsilon \Big(\Psi^{t_1,\epsilon} \,\Theta(\xi^1)\,\cdots \, 
\Psi^{t_m,\epsilon} \,\Theta(\xi^m)\Big)  \;=\;
\tilde{\rho}_1 \Big(\Psi^{t_1} \,\Theta(\xi^1)\,\cdots \, 
\Psi^{t_m} \,\Theta(\xi^m)\Big)\,,
\end{equation}
uniformly in \eqref{uniformity_k_epsilon}.
In order to prove \eqref{local_nonlinearity_claim2}, we note that, by construction of $\Psi^{t,\epsilon},\Psi^t$ and Proposition \ref{L2 convergence}, we have that, for $\xi \in \cal L(\fra H^{(p)})$, 
\begin{equation}
\label{local_nonlinearity_auxiliary_claim}
\lim_{\epsilon \to 0} \Psi^{t,\epsilon} \,\Theta(\xi) \;=\; \Psi^t \,\Theta(\xi) \quad \mbox{in } \fra H \mbox{ almost surely.}
\end{equation}
The convergence in \eqref{local_nonlinearity_auxiliary_claim} is uniform in $\abs{t} \leq k, p_i \leq k, \|\xi^i\| \leq k$.
Indeed, we write
\begin{equation*}
\Psi^{t,\epsilon} \,\Theta(\xi) \;=\; \Big\langle (S_t^{\epsilon} \phi)^{\otimes k}, \xi \,(S_t^{\epsilon} \phi)^{\otimes k}\Big\rangle_{\fra H^{\otimes k}}\,,\quad  \Psi^{t} \,\Theta(\xi) \;=\; \Big\langle (S_t \phi)^{\otimes k}, \xi \,(S_t \phi)^{\otimes k}\Big\rangle_{\fra H^{\otimes k}}\,,
\end{equation*}
where $S_{t}^{\epsilon}$ and $S_t$ denote the flow maps of \eqref{NLS_epsilon} and \eqref{NLS} respectively. We consider the initial data $\phi_0$ given by the free classical field $\phi$ \eqref{classical_free_field}.
Let us recall that $\phi \in H^s \subset \fra H$ almost surely.
Proposition \ref{L2 convergence} then implies that $\lim_{\epsilon \to 0} (S_t^{\epsilon} \phi)^{\otimes k} = (S_t \phi)^{\otimes k}$ in $\fra H^{\otimes k}$ almost surely. We deduce \eqref{local_nonlinearity_auxiliary_claim} since $\xi \in \cal L(\fra H^{(p)})$.

In particular, from \eqref{W convergence} and \eqref{local_nonlinearity_auxiliary_claim} it follows that 
\begin{equation}
\label{local_nonlinearity_claim2A}
\lim_{\epsilon \to 0} \Psi^{t_1,\epsilon} \,\Theta(\xi^1)\,\cdots \, 
\Psi^{t_m,\epsilon} \,\Theta(\xi^m)\,\ee^{-\cal W^\epsilon} \;=\;
\Psi^{t_1} \,\Theta(\xi^1)\,\cdots \, 
\Psi^{t_m} \,\Theta(\xi^m)\,\ee^{-\cal W} \quad \mbox{almost surely.}
\end{equation}
Furthermore, by conservation of mass for \eqref{NLS_epsilon} and since $\cal W^\epsilon \geq 0$ by construction, it follows that for all $\epsilon>0$ we have
\begin{equation}
\label{local_nonlinearity_claim2B}
\Big|\Psi^{t_1,\epsilon} \,\Theta(\xi^1)\,\cdots \, 
\Psi^{t_m,\epsilon} \,\Theta(\xi^m)\,\ee^{-\cal W^\epsilon}\Big| \;\leq\; \|\xi^1\| \, \cdots \, \|\xi^m\| \,\|\phi\|_{\fra H}^{\,2(p_1+\cdots+p_m)} \in L^1(\dd \mu)\,.
\end{equation}
We now deduce \eqref{local_nonlinearity_claim2} from \eqref{local_nonlinearity_claim2A}--\eqref{local_nonlinearity_claim2B} and the dominated convergence theorem.
\end{proof}

\appendix

\section{$X^{\sigma,b}$ spaces: proof of Lemma \ref{Xsb space lemma}.}
\label{Xsb space appendix}

In this appendix we present the proof of Lemma \ref{Xsb space lemma}.
We emphasize that this is done for the convenience of the reader and that it is not an original contribution of the paper.

\begin{proof}[Proof of Lemma \ref{Xsb space lemma}.]
We recall the definition of $b$ given in \eqref{definition of b}. 

We first prove part (i). The proof is analogous to the proof the Sobolev embedding $H^b_t \hookrightarrow L^{\infty}_t$. We use the Fourier inversion formula in the time variable and write
\begin{equation}
\label{Fourier_inversion_t}
\hat{f}(k,t) \;=\; \int_{-\infty}^{\infty} \dd \eta \,\tilde{f}(k,\eta)\,\ee^{2\pi \ii \eta t}\,.
\end{equation}
In \eqref{Fourier_inversion_t}, 
\begin{equation*}
\hat{f}(k,t) \;=\;\int_{\Lambda} \dd x\, f(x,t)\,\ee^{-2\pi \ii k x}
\end{equation*}
denotes the Fourier transform in the space variable.
In particular, using the Cauchy-Schwarz inequality in $\eta$ in \eqref{Fourier_inversion_t} and recalling that $b>1/2$, it follows that
\begin{equation}
\label{hat_f_bound}
\big|\hat{f}(k,t)\big| \;\leq\; C(b) \Bigg(\int_{-\infty}^{\infty} \dd \eta\,|\tilde{f}(k,\eta)|^2 \, \big(1+|\eta+2\pi k^2|\big)^{2b}\Bigg)^{1/2}\,.
\end{equation}
Claim (i) follows from \eqref{hat_f_bound} and Definition \ref{X^{sigma,b}}.

Claims analogous to (ii)-(iv) were proved for $X^{\sigma,b}$ corresponding to the Airy equation in the non-periodic setting \cite[Lemmas 3.1--3.3]{Kenig_Ponce_Vega}. The bounds for the Schr\"{o}dinger equation follow in the same way, since we are estimating integrals in the Fourier variable $\eta$ dual to time. For completeness, we give the proofs of (ii)-(iv).

We proceed with the proof of (ii). By density, it suffices to consider $\Phi \in \cal S(\Lambda_x)$. Let us note that, for fixed $x \in \Lambda$
\begin{equation*}
\psi(t/\delta)\,\ee^{\ii t \Delta}\Phi \;=\;\psi(t/\delta) \,\sum_k \ee^{2\pi i kx-4\pi^2 \ii k^2 t} \,\hat{\Phi}(k)\,,
\end{equation*}
from where we deduce that
\begin{equation*}
\big(\psi(t/\delta) \,\ee^{\ii t \Delta}\Phi\big) \,\widetilde{}\,(k,\tau)\;=\;\delta \,\hat{\psi}\Big(\delta(\eta+2\pi k^2)\Big) \,\hat{\Phi}(k)\,.
\end{equation*}
Hence
\begin{multline}
\label{free_evolution_Xsb_bound}
\big\|\psi(t/\delta)\,\ee^{\ii t \Delta}\Phi \big\|_{X^{\sigma,b}}^2
\;=\; \sum_k \big(1+|2\pi k|\big)^{2\sigma}\,|\hat{\Phi}(k)|^2 \, \bigg[\delta^2 \, \int_{-\infty}^{\infty} \dd \eta \,  \,\Big|\hat{\psi}\Big(\delta(\eta+2\pi k^2)\Big)\Big|^2 \,\big(1+|\eta+ 2\pi k^2|\big)^{2b}\bigg]
\\
\;=\;\|\Phi\|_{H^{\sigma}}^2 \,  \bigg[\delta^2 \, \int_{-\infty}^{\infty} \dd \eta \,  \,\big|\hat{\psi}(\delta \eta)\big|^2 \,\big(1+|\eta|\big)^{2b}\bigg]\,.
\end{multline}
By scaling we obtain that the following estimates hold.
\begin{equation}
\label{free_evolution_Xsb_bound1}
\delta^2 \, \int_{-\infty}^{\infty} \dd \eta \,  \,\big|\hat{\psi}(\delta \eta)\big|^2 \;\leq\; C(\psi)\,.
\end{equation}
\begin{equation}
\label{free_evolution_Xsb_bound2}
\delta^2 \, \int_{-\infty}^{\infty} \dd \eta \,  \,\big|\hat{\psi}(\delta \eta)\big|^2 \,|\eta|^{2b}\;\leq\; C(b,\psi)\,\delta^{1-2b}\,.
\end{equation}
Claim (ii) follows by substituting \eqref{free_evolution_Xsb_bound1}--\eqref{free_evolution_Xsb_bound2} into \eqref{free_evolution_Xsb_bound} and using
\begin{equation*}
\big(1+|\eta|\big)^{2b} \;\leq\; C(b)\big(1+|\eta|^{2b}\big)\,.
\end{equation*}
(In the sequel, we use the latter elementary inequality repeatedly without explicit mention).

We now prove (iii). Let us note that 
\begin{equation}
\label{Xsb_Lemma_iii1}
\big(\psi(t/\delta)f\big)\,\,\widetilde{}\,\,(k,\eta)\;=\;\tilde{f}(k,\eta) *_{\eta}\Big(\delta \,\hat{\psi}(\delta \cdot) \Big)\,
\end{equation}
where $*_\eta$ denotes convolution in $\eta$.
From \eqref{Xsb_Lemma_iii1} and Definition \ref{X^{sigma,b}}, it follows that (iii) is equivalent to showing that for all $h=h(t)$ and $a \in \R$ we have
\begin{equation}
\label{Xsb_Lemma_iii2}
\int_{-\infty}^{\infty} \dd \eta \, \Big| \hat{h} *_\eta \Big(\delta \hat{\psi}(\delta \cdot) \Big)(\eta)\Big|^2\,\big(1+|\eta+a|\big)^{2b}
\;\leq\;C(b,\psi) \, \delta^{1-2b} \, \int_{-\infty}^{\infty} \dd \eta \, |\hat{h}(\eta)|^2\,\big(1+|\eta+a|\big)^{2b}\,.
\end{equation}
By Young's inequality, it follows that
\begin{equation}
\label{Xsb_Lemma_iii3}
\int_{-\infty}^{\infty} \dd \eta \, \Big| \hat{h} *_\eta \Big(\delta \hat{\psi}(\delta \cdot) \Big)(\eta)\Big|^2 \;\leq\; C(\psi) \int_{-\infty}^{\infty} \dd \eta \, |\hat{h}(\eta)|^2\,.
\end{equation}
Moreover, we write
\begin{multline}
\label{Xsb_Lemma_iii4}
\int_{-\infty}^{\infty} \dd \eta \,  \Big| \hat{h} *_\eta \Big(\delta \hat{\psi}(\delta \cdot) \Big)(\eta)\Big|^2\,|\eta+a|^{2b} \;=\;\int_{-\infty}^{\infty} \dd \eta \,  \Big| \hat{h} *_\eta \Big(\delta \hat{\psi}(\delta \cdot) \Big)(\eta-a)\Big|^2\,|\eta|^{2b}
\\
\;=\; C(b) \int_{-\infty}^{\infty} \dd t \,\Big| |\partial|^b \Big(\ee^{2\pi i a t} \,h(t)\,\psi(\delta^{-1}\,t)  \Big) \Big|^2 \;=\; C(b) \Big\||\partial|^b \Big(\ee^{2\pi a t}\,h\,\psi(\delta^{-}\cdot)\Big)\Big\|_{L^2_t}^2\,.
\end{multline}
Here we use the notation $|\partial|^b$ for the fractional differentiation operator given by 
\begin{equation*}
(|\partial|^b g)\,\,\widehat{}\,\,(\eta)\;=\;|2 \pi \eta|^b \,\hat{g}(\eta)\,.
\end{equation*}
We now refer to the result of \cite[Theorem A.12]{Kenig_Ponce_Vega1} (c.f.\ also \cite[Theorem 2.8]{Kenig_Ponce_Vega}) which states that for all $\alpha \in (0,1)$ and $p \in (1,\infty)$ we have
\begin{equation}
\label{Fractional_Leibniz_rule_bound1}
\big\||\partial|^{\alpha}(fg)-f\,|\partial|^{\alpha}g\big\|_{L^p} \;\leq C(\alpha,p) \,\|g\|_{L^{\infty}} \, \big\||\partial|^{\alpha}f\big\|_{L^p}\,.
\end{equation}
Taking $\alpha=b$, $p=2$, $f=\ee^{2\pi \ii a t} \, h$ and $g=\psi(\delta^{-1}\cdot)$ in \eqref{Fractional_Leibniz_rule_bound1} we obtain
\begin{multline}
\label{Fractional_Leibniz_rule_bound}
\Big\||\partial|^b \Big(\ee^{2\pi \ii a t} \, h \, \psi(\delta^{-1}\cdot)\Big)-\ee^{2\pi \ii a t} \,h\,|\partial|^b \Big(\psi(\delta^{-1}\cdot)\Big)\Big\|_{L^2_t} \;\leq\;C(b) \|\psi(\delta^{-1}\cdot)\|_{L^\infty_t} \, \big\||\partial|^b (\ee^{2\pi \ii a t} h)\big\|_{L^2_t} 
\\
\;\leq\; C(b,\psi) \, \big\||\partial|^b (\ee^{2\pi \ii a t} h)\big\|_{L^2_t}\,.
\end{multline}
By Plancherel's theorem we have
\begin{equation}
\label{Xsb_Lemma_iii5}
\big\||\partial|^b (\ee^{2\pi \ii a t} h)\big\|_{L^2_t}^2 \;=\;C(b) \int_{-\infty}^{\infty} \dd \eta \, |\hat{h}(\eta-a)|^2\, |\eta|^{2b}  \;=\;C(b) \int_{-\infty}^{\infty} \dd \eta \, |\hat{h}(\eta)|^2\, |\eta+a|^{2b}\,.
\end{equation}
From \eqref{Xsb_Lemma_iii3}--\eqref{Xsb_Lemma_iii4} and \eqref{Fractional_Leibniz_rule_bound}--\eqref{Xsb_Lemma_iii5}, we deduce that \eqref{Xsb_Lemma_iii2} follows if we show that 
\begin{equation}
\label{Xsb_Lemma_iii6}
I \;\deq\; \Big\|\ee^{2\pi \ii a t} \,h\,|\partial|^b \Big(\psi(\delta^{-1}\cdot)\Big)\Big\|_{L^2_t}
\;\leq\;C(b,\psi) \, \delta^{\frac{1-2b}{2}} \, \bigg(\int_{-\infty}^{\infty} \dd \eta \, |\hat{h}(\eta)|^2\,\big(1+|\eta+a|\big)^{2b}\bigg)^{1/2}\,.
\end{equation}
Applying H\"{o}lder's inequality and Sobolev embedding with $b>\frac{1}{2}$, it follows that 
\begin{multline}
\label{Xsb_Lemma_iii7}
I \;\leq\; \big\|\ee^{2\pi \ii a t}h\big\|_{L^{\infty}_t} \, \Big\||\partial|^b \Big(\psi(\delta^{-1}\cdot)\Big)\Big\|_{L^2_t} \;\leq\; C(b) \, \big\|\ee^{2\pi \ii a t}h\big\|_{H^{b}_t} \, \Big\||\partial|^b \Big(\psi(\delta^{-1}\cdot)\Big)\Big\|_{L^2_t}
\\
\;=\; C(b)\, \bigg(\int_{-\infty}^{\infty} \dd \eta \, |\hat{h}(\eta)|^2\,\big(1+|\eta+a|\big)^{2b}\bigg)^{1/2}\, \Big\||\partial|^b \Big(\psi(\delta^{-1}\cdot)\Big)\Big\|_{L^2_t} \,.
\end{multline}
By scaling, we compute 
\begin{equation}
\label{Xsb_Lemma_iii8}
\Big\||\partial|^b \Big(\psi(\delta^{-1}\cdot)\Big)\Big\|_{L^2_t}
\;=\; C(b)\, \delta^{\frac{1-2b}{2}}\,\bigg(\int_{-\infty}^{\infty} \dd \eta\, |\eta|^{2b} \, |\hat{\psi}(\eta)|^2\bigg)^{1/2} \;\leq\; C(b,\psi)\,\delta^{\frac{1-2b}{2}}\,.
\end{equation}
By \eqref{Xsb_Lemma_iii7}--\eqref{Xsb_Lemma_iii8}, we deduce \eqref{Xsb_Lemma_iii6}, which in turn implies \eqref{Xsb_Lemma_iii2}. The claim (iii) now follows.

We now prove (iv). By density, it suffices to consider $f \in \cal S(\Lambda_x \times \mathbb{R}_t)$. 
We write
\begin{equation}
\label{J_definition}
J \;\deq\; \psi(t/\delta)\,\int_0^t \dd t'\,\ee^{\ii (t-t') \Delta}\,f(t') 
\;=\;\psi(t/\delta)\,\int_0^t \dd t'\,\int_{-\infty}^{\infty} \dd \eta\, \sum_k \tilde{f}(k,\eta) \,\ee^{2\pi \ii k x} \, \ee^{-4\pi^2 \ii k^2 t} \, \ee^{2\pi \ii \eta t'(\eta+2\pi k^2)}\,.
\end{equation} 
By the assumptions on $f$ we can interchange the orders of integration so that we first integrate in $t'$. Evaluating the $t'$ integral, it follows that
\begin{equation}
\label{I_1+I_2}
J\;=\;\psi(t/\delta)\,\int_{-\infty}^{\infty} \dd \eta\, \sum_k \tilde{f}(k,\eta) \,\ee^{2\pi \ii k x} \, \frac{\ee^{2\pi \ii \eta t}-\ee^{-4\pi^2 \ii k^2t}}{2\pi \ii (\eta+2\pi k^2)}\;=\;I_1+I_2\,,
\end{equation}
where
\begin{equation}
\label{I_1}
I_1 \;\deq\;\psi(t/\delta)\,\int_{-\infty}^{\infty} \dd \eta\, \sum_k \tilde{f}(k,\eta) \,\ee^{2\pi \ii k x} \, \Xi(\eta+2\pi k^2)\,\frac{\ee^{2\pi \ii \eta t}-\ee^{-4\pi^2 \ii k^2t}}{2\pi \ii (\eta+2\pi k^2)}
\end{equation}
\begin{equation}
\label{I_2}
I_2 \;\deq\;\psi(t/\delta)\,\int_{-\infty}^{\infty} \dd \eta\, \sum_k \tilde{f}(k,\eta) \,\ee^{2\pi \ii k x} \, \Big(1-\Xi(\eta+2\pi k^2)\Big)\,\frac{\ee^{2\pi \ii \eta t}-\ee^{-4\pi^2 \ii k^2t}}{2\pi \ii (\eta+2\pi k^2)}\,,
\end{equation}
for a function $\Xi \in C_c^{\infty}(\mathbb{R})$ such that 

\begin{equation}
\label{assumptions on Xi}
\Xi\;=\;1\,\quad \mbox{for} \quad |y| \;\leq\; 1/2\,, \quad \mbox{and} \quad 
\Xi\;=\;0\,\quad \mbox{for} \quad |y| \;>\; 1\,.
\end{equation}

We first consider $I_1$. By writing a Taylor expansion for the factor 
\begin{equation*}
\frac{\ee^{2\pi \ii \eta t}-\ee^{-4\pi^2 \ii k^2t}}{2\pi \ii (\eta+2\pi k^2)} \;=\;\ee^{-4\pi^2\ii k^2 t} \,\frac{\ee^{2\pi \ii t (\eta+2\pi k^2)}-1}{2\pi \ii (\eta+2\pi k^2)}
\end{equation*}
in the integrand of \eqref{I_1}, we have
\begin{equation}
\label{I_1A}
I_1\;=\;\sum_{l=1}^{\infty} \frac{(2\pi \ii)^{l-1}\,t^l}{l!} \, \psi(t/\delta) \, \int_{-\infty}^{\infty} \dd \eta\,\sum_k \tilde{f}(k,\eta)\,\ee^{2\pi \ii k x} \, \Xi(\eta+2\pi k^2)\,\ee^{-4\pi^2 \ii k^2 t}\, (\eta+2\pi k^2)^{l-1}\,.
\end{equation}
Here, we can justify taking the sum in $l$ outside of the integral in $\eta$ and the sum in $k$ using the assumption that $f \in \cal S(\Lambda_x \times \mathbb{R}_t)$ as before. Setting 
\begin{equation}
\label{psi_l_definition}
\psi_l(y) \;\deq\; y^l \psi(y)\,
\end{equation}
and using the Fourier representation of $\ee^{\ii t \Delta}$, we can rewrite \eqref{I_1A} as
\begin{equation}
\label{I_1B}
I_1 \;=\; \sum_{l=1}^{\infty}  \frac{(2\pi \ii)^{l-1}\,\delta^l}{l!} \, \psi_l(t/\delta) \,\bigg[\ee^{\ii t \Delta} \, \bigg( \int_{-\infty}^{\infty} \dd \eta\,\sum_k \tilde{f}(k,\eta)\,\ee^{2\pi \ii k x} \, \Xi(\eta+2\pi k^2)\,(\eta+2\pi k^2)^{l-1}\bigg)\bigg]\,.
\end{equation}
Using the triangle inequality and \eqref{free_evolution_Xsb_bound}, we obtain from \eqref{I_1B} that
\begin{multline}
\label{I_1C}
\|I_1\|_{X^{\sigma,b}} \;\leq\; \sum_{l=1}^{\infty}  \Bigg(\frac{(2\pi)^{l-1}\,\delta^l}{l!} \,  \bigg[\delta^2 \, \int_{-\infty}^{\infty} \dd \eta \,  \,\big|\hat{\psi}_l(\delta \eta)\big|^2 \,(1+|\eta|)^{2b}\bigg]^{1/2} 
\\
\times
\bigg\|\bigg( \int_{-\infty}^{\infty} \dd \eta\,\sum_k \tilde{f}(k,\eta)\,\ee^{2\pi \ii k x} \, \Xi(\eta+2\pi k^2)\,(\eta+2\pi k^2)^{l-1}\bigg)\bigg\|_{H^{\sigma}}\Bigg)\,.
\end{multline}
By scaling, we compute
\begin{equation}
\label{psi_l_A}
\delta^2 \, \int_{-\infty}^{\infty} \dd \eta \,  \,\big|\hat{\psi}_l(\delta \eta)\big|^2 \;=\;\delta \|\psi_l\|_{L^2}^2\,.
\end{equation}
\begin{equation}
\label{psi_l_B}
\delta^2 \, \int_{-\infty}^{\infty} \dd \eta \,  \,\big|\hat{\psi}_l(\delta \eta)\big|^2\,|\eta|^{2b} \;=\;\delta^{1-2b} \|\psi_l\|_{\dot{H}^b}^2\,.
\end{equation}
In particular, from \eqref{psi_l_A}--\eqref{psi_l_B}, we deduce that 
\begin{multline}
\label{psi_l}
\delta^2 \, \int_{-\infty}^{\infty} \dd \eta \,  \,\big|\hat{\psi}_l(\delta \eta)\big|^2 \,(1+|\eta|)^{2b} \;\leq\; C(b) \,\delta^{1-2b}\, \Big(\|\psi_l\|_{L^2}^2+\|\psi_l\|_{\dot{H}^b}^2\Big) \;\leq\; C(b) \,\delta^{1-2b}\, \Big(\|\psi_l\|_{L^2}^2+\|\psi_l\|_{\dot{H}^1}^2\Big)
\\
\;\leq\;
C(b) \,\delta^{1-2b}\, \Big(\|y^l \psi\|_{L^2}^2+\|y^l \psi'\|_{L^2}^2 + \|ly^{l-1}\psi\|_{L^2}^2\Big)\;\leq\; C(b) \,\delta^{1-2b}\,\big(C(\psi)\big)^l \,.
\end{multline} 
Above we used \eqref{psi_l_definition} and the assumption that $\psi \in C_c^{\infty}(\mathbb{R})$.

Moreover, we have 
\begin{multline*}
\bigg\|\bigg( \int_{-\infty}^{\infty} \dd \eta\,\sum_k \tilde{f}(k,\eta)\,\ee^{2\pi \ii k x} \, \Xi(\eta+2\pi k^2)\,(\eta+2\pi k^2)^{l-1}\bigg)\bigg\|_{H^{\sigma}}^2
\\
\;=\; \sum_k (1+|2 \pi k|)^{2\sigma}\,\Bigg| \int_{-\infty}^{\infty} \dd \eta\,\sum_k \tilde{f}(k,\eta)\,\Xi(\eta+2\pi k^2)\,(\eta+2\pi k^2)^{l-1}\Bigg|^2 
\\
\;\leq\; C(\Xi)\,\sum_k (1+|2\pi k|)^{2\sigma} \, \Bigg(\int_{|\eta+2\pi k^2| \;\leq\;1} \dd \eta\,|\tilde{f}(k,\eta)|\Bigg)^2\,,
\end{multline*}
uniformly in $l$. In the above line, we used the assumptions \eqref{assumptions on Xi} on the function $\Xi$. In particular, the above expression is 
\begin{multline}
\label{I_1D}
\;\leq\; C(\Xi)\,\sum_k (1+|2 \pi k|)^{2\sigma} \, \Bigg(\int_{|\eta+2\pi k^2| \;\leq\;1} \dd \eta\,\frac{|\tilde{f}(k,\eta)|}{1+|\eta+2\pi k^2|}\Bigg)^2 
\\
\;\leq\; C(\Xi)\,\sum_k (1+|2 \pi k|)^{2\sigma} \, \Bigg(\int_{-\infty}^{\infty} \dd \eta\,\frac{|\tilde{f}(k,\eta)|}{(1+|\eta+2\pi k^2|)^{1-b}} \, \frac{1}{(1+|\eta+2\pi k^2|)^b}\Bigg)^2
\\
\;\leq\; C(b,\Xi) \,\Bigg(\int_{-\infty}^{\infty} \dd \eta\,\sum_k  (1+|2 \pi k|)^{2\sigma} \,(1+|\eta+2\pi k^2|)^{2(b-1)}\,|\tilde{f}(k,\eta)|^2 \Bigg) \;=\;C(b,\Xi)\,\|f\|_{X^{s,b-1}}^2\,.
\end{multline}
In the last inequality we used the Cauchy-Schwarz inequality in $\eta$ and the assumption that $b>\frac{1}{2}$. Combining \eqref{I_1C}, \eqref{psi_l}--\eqref{I_1D}, it follows that 
\begin{equation}
\label{I_1_bound}
\|I_1\|_{X^{\sigma,b}} \;\leq\; C(b,\psi) \, \delta^{\frac{1-2b}{2}}\,\|f\|_{X^{\sigma,b-1}}\,.
\end{equation}
(Since $\Xi$ is chosen as an arbitrary $C_c^\infty$ function satisfying \eqref{assumptions on Xi}, we do not keep track of the dependence of the implied constant on this function).

We now consider $I_2$ as defined in \eqref{I_2}. We write
\begin{equation}
\label{I_21-I_22}
I_2 \;=\;I_{2,1}-I_{2,2}\,,
\end{equation}
where 
\begin{equation*}
I_{2,1} \;\deq\; \psi(t/\delta)\,\int_{-\infty}^{\infty} \dd \eta\, \sum_k \tilde{f}(k,\eta) \,\ee^{2\pi \ii k x} \, \Big(1-\Xi(\eta+2\pi k^2)\Big)\,\frac{\ee^{2\pi \ii \eta t}}{2\pi \ii (\eta+2\pi k^2)}\,.
\end{equation*}
\begin{equation*}
I_{2,2} \;\deq\; \psi(t/\delta)\,\int_{-\infty}^{\infty} \dd \eta\, \sum_k \tilde{f}(k,\eta) \,\ee^{2\pi \ii k x} \, \Big(1-\Xi(\eta+2\pi k^2)\Big)\,\frac{\ee^{-4\pi^2 \ii k^2t}}{2\pi \ii (\eta+2\pi k^2)}\,.
\end{equation*}
We first estimate $I_{2,1}$. By claim (iii), we have
\begin{equation}
\label{I_21_boundA}
\|I_{2,1}\|_{X^{\sigma,b}} \;\leq\;C(b,\psi)\,\delta^{\frac{1-2b}{2}}\,\bigg\|\int_{-\infty}^{\infty} \dd \eta\, \sum_k \tilde{f}(k,\eta)\,\frac{\big(1-\Xi(\eta+2\pi k^2)\big)}{\eta+2\pi k^2} \,\ee^{2\pi \ii k x + 2\pi \ii \eta t} \bigg\|_{X^{\sigma,b}}\,.
\end{equation}
Note that, by the assumptions \eqref{assumptions on Xi} on $\Xi$ we have
\begin{equation}
\label{I_21_boundB}
\bigg|\tilde{f}(k,\eta)\,\frac{\big(1-\Xi(\eta+2\pi k^2)\big)}{\eta+2\pi k^2}\bigg| \;\leq\; C(\Xi) \, \frac{|\tilde{f}(k,\eta)|}{\big(1+|\eta+2\pi k^2|\big)}\,.
\end{equation}
Combining \eqref{I_21_boundA}--\eqref{I_21_boundB} and recalling the definition of $\|\cdot\|_{X^{\sigma,b}}$ we deduce that
\begin{multline}
\label{I_21_bound}
\|I_{2,1}\|_{X^{\sigma,b}} \;\leq\;C(b,\psi,\Xi) \, \delta^{\frac{1-2b}{2}}\,
\Big\|\big(1+|2\pi k|\big)^{\sigma}\,\big(1+|\eta+2\pi k^2|\big)^{b-1}\, |\tilde{f}(k,\eta)|\Big\|_{L^2_{\eta} l^2_k} 
\\
\;=\;C(b,\psi,\Xi) \, \delta^{\frac{1-2b}{2}}\, \|f\|_{X^{\sigma,b-1}}\,.
\end{multline}
In particular, in the last line, we used that the $X^{\sigma,b}$ norm depends only on the absolute value of the spacetime Fourier transform.

We now estimate $I_{2,2}$. Let us note that
\begin{equation}
\label{I_22_A}
I_{2,2}\;=\;\psi(t/\delta)\,\ee^{\ii t \Delta} \, \Bigg[\sum_k \Bigg(\int_{-\infty}^{\infty} \dd \eta\,\tilde{f}(k,\eta)  \, \frac{\big(1-\Xi(\eta+2\pi k^2)\big)}{2\pi \ii (\eta+2\pi k^2)} \Bigg)\,\ee^{2\pi \ii k x}\Bigg]\,.
\end{equation}
From \eqref{I_22_A} and part (ii) we obtain
\begin{equation*}
\|I_{2,2}\|_{X^{\sigma,b}} \;\leq\; C(b,\psi)\,\delta^{\frac{1-2b}{2}}\,\Bigg\|\sum_k \Bigg(\int_{-\infty}^{\infty} \dd \eta\,\tilde{f}(k,\eta)  \, \frac{\big(1-\Xi(\eta+2\pi k^2)\big)}{2\pi \ii (\eta+2\pi k^2)} \Bigg)\,\ee^{2\pi \ii k x}\Bigg\|_{H^{\sigma}},
\end{equation*}
which, by using the definition of $\|\cdot\|_{H^\sigma}$ and \eqref{I_21_boundB} is
\begin{equation}
\label{I_22_B}
\;\leq\; C(b,\psi,\Xi) \, \delta^{\frac{1-2b}{2}}\,\Bigg[\sum_k \big(1+ |2\pi k|\big)^{2\sigma} \,\Bigg(\int_{-\infty}^{\infty} \dd \eta\,\frac{|\tilde{f}(k,\eta)|}{\big(1+|\eta+2\pi k^2|\big)} \Bigg)^2 \Bigg]^{1/2}\,.
\end{equation}
By writing 
\begin{equation*}
\frac{1}{\big(1+|\eta+2\pi k^2|\big)}\;=\;\frac{1}{\big(1+|\eta+2\pi k^2|\big)^{1-b}}\,\frac{1}{\big(1+|\eta+2\pi k^2|\big)^b}
\end{equation*}
in the integrand in \eqref{I_22_B} and applying the Cauchy-Schwarz inequality in $\eta$ analogously as in the proof of \eqref{I_1D} above, it follows that 
\begin{multline}
\label{I_22_bound}
\|I_{2,2}\|_{X^{\sigma,b}} \;\leq\; C(b,\psi,\Xi) \, \delta^{\frac{1-2b}{2}}\,\Bigg[\sum_k \big(1+ |2\pi k|\big)^{2\sigma} \,\int_{-\infty}^{\infty} \dd \eta\,\frac{|\tilde{f}(k,\eta)|^2}{\big(1+|\eta+2\pi k^2|\big)^{2(1-b)}}  \Bigg]^{1/2}
\\
\;=\; C(b,\psi,\Xi) \, \delta^{\frac{1-2b}{2}}\, \|f\|_{X^{\sigma,b-1}}\,.
\end{multline}
From \eqref{I_21-I_22}, \eqref{I_21_bound} and \eqref{I_22_bound}, we obtain
\begin{equation}
\label{I_2_bound}
\|I_2\|_{X^{\sigma,b}} \;\leq\; C(b,\psi) \, \delta^{\frac{1-2b}{2}}\,\|f\|_{X^{\sigma,b-1}}\,.
\end{equation}
(As in \eqref{I_1_bound}, we do not emphasize the $\Xi$-dependence in the implied constant).
Claim (iv) now follows from \eqref{J_definition}--\eqref{I_1+I_2}, \eqref{I_1_bound} and \eqref{I_2_bound}.

Claim (v) is proved in \cite[Proposition 2.6]{Bourgain_1993}. For an alternative proof, see also \cite[Proposition 2.13]{Tao}. Claim (vi) follows from part (v) by duality. 

Finally, we prove claim (vii).
By part (v), it suffices to prove 
\begin{equation}
\label{Xsb space lemma Claim 7}
\big\|\psi(t/\delta)f\big\|_{X^{0,3/8}} \;\leq\;C(b)\,\delta^{\theta_0}\,\|f\|_{X^{0,b}}\,
\end{equation}
for appropriately chosen $\theta_0>0$.
In order to prove claim \eqref{Xsb space lemma Claim 7} we first note that 
\begin{equation}
\label{Xsb interpolation 1}
X^{0,\frac{1}{4}+} \hookrightarrow L^4_t L^2_x\,,
\end{equation}
 which follows by using $X^{0,0}=L^2_{t,x}$, part (i) with $\sigma=0$ and interpolation. By an additional interpolation step, it follows that
\begin{equation}
\label{Xsb interpolation 2}
\big\|\psi(t/\delta)f\big\|_{X^{0,3/8}} \;\leq\; \big\|\psi(t/\delta)f\big\|_{X^{0,0}}^\theta \, \big\|\psi(t/\delta)f\big\|_{X^{0,b}}^{1-\theta} \quad \mbox{for} \quad \theta \deq \frac{b-\frac{3}{8}}{b} \;=\;\frac{1}{4}+\,.
\end{equation}
We have, by H\"{o}lder's inequality and \eqref{Xsb interpolation 1}
\begin{equation}
\label{Xsb interpolation 3}
\big\|\psi(t/\delta)f\big\|_{X^{0,0}} \;=\;\big\|\psi(t/\delta)f\big\|_{L^2_{t,x}} \;\leq\; \|\psi(t/\delta)\|_{L^4_t} \, \|f\|_{L^4_tL^2_x} \;\leq\; C\delta^{\frac{1}{4}} \,\|f\|_{X^{0,\frac{1}{4}+}} \;\leq\; C\delta^{\frac{1}{4}} \,\|f\|_{X^{0,b}}\,. 
\end{equation}
From \eqref{Xsb interpolation 2}-\eqref{Xsb interpolation 3} and part (iii), we obtain claim (vii) with 
\begin{equation}
\label{theta_0 definition}
\theta_0\;\deq\;\frac{\theta}{4}+(1-\theta)\frac{1-2b}{2} \;>\;0\,.
\end{equation}
\end{proof}

\end{document}